\documentclass[11pt,a4paper]{article}
\pdfoutput=1

\usepackage[utf8]{inputenc}
\usepackage{authblk}

\title{\bfseries Beyond Highway Dimension:\\ Small Distance Labels Using Tree Skeletons\thanks{Supported by Inria project GANG, ANR project DESCARTES, and NCN grant 2015/17/B/ST6/01897.}}

\author{Adrian Kosowski}
\author{Laurent Viennot}
\affil{Inria Paris and IRIF, Universit\'e Paris Diderot, France}

\date{}

\usepackage{verbatim}
\usepackage{makecmds}

\usepackage{times}
\usepackage{fullpage}
\usepackage{amsthm}
\usepackage{color}
\usepackage{graphicx}
\usepackage{xspace}
\usepackage{cite}
\usepackage{url}
\usepackage[hidelinks]{hyperref}
\hypersetup{pdfstartview=XYZ}
\usepackage{todonotes} 

\def\laurent#1{}
\def\adrian#1{}

\long\def\jump#1\finjump{}

\newtheorem{theorem}{Theorem}
\newtheorem{lemma}{Lemma}
\newtheorem{corollary}{Corollary}
\newtheorem{proposition}{Proposition}
\newtheorem{claim}{Claim}

\newcommand{\card}[1]{\left|{#1}\right|}
\newcommand{\floor}[1]{\left\lfloor{#1}\right\rfloor}
\newcommand{\ceil}[1]{\left\lceil{#1}\right\rceil}

\newcommand{\set}[1]{\left\{{#1}\right\}}

\newcommand{\bigo}[1]{O\mathopen{}\left(#1\right)}

\newcommand{\E}{\mathbb{E}}
\renewcommand{\P}{\mathcal{P}}

\newcommand{\wtilde}[1]{\widetilde{#1}}
\newcommand{\G}{{\wtilde{G}}}
\newcommand{\T}{{\wtilde{T}}}

\let\eps=\varepsilon

\let\DD=\delta
\let\cost=d
\let\symb=x
\let\a=\alpha
\let\b=\beta

\usepackage{amsfonts}
\newcommand{\R}{\mathbb{R}}
\newcommand{\N}{\mathbb{N}}

\usepackage{amsmath}
\DeclareMathOperator{\reach}{\mathit{Reach}}
\DeclareMathOperator{\width}{\mathit{Width}}
\DeclareMathOperator{\sk}{\mathit{k}}
\DeclareMathOperator{\isk}{\mathit{\hat{k}}}
\DeclareMathOperator{\diam}{\mathit{D}}
\DeclareMathOperator{\cut}{\mathit{Cut}}
\DeclareMathOperator{\dist}{\delta}

\DeclareMathOperator{\aspect}{\diam}
\DeclareMathOperator{\rand}{rand}

\begin{document}

\maketitle

\begin{abstract}
The goal of a hub-based distance labeling scheme for a network $G = (V,E)$ is to assign a small subset $S(u) \subseteq V$ to each node $u \in V$, in such a way that for any pair of nodes $u, v$, the intersection of hub sets $S(u) \cap S(v)$ contains a node on the shortest $uv$-path.

The existence of small hub sets, and consequently efficient shortest path processing algorithms, for road networks is an empirical observation. A theoretical explanation for this phenomenon was proposed by Abraham et al.\ (SODA 2010) through a network parameter they called \emph{highway dimension}, which captures the size of a hitting set for a collection of shortest paths of length at least $r$ intersecting a given ball of radius $2r$. In this work, we revisit this explanation, introducing a more tractable (and directly comparable) parameter based solely on the structure of shortest-path spanning trees, which we call \emph{skeleton dimension}. We show that skeleton dimension admits an intuitive definition for both directed and undirected graphs, provides a way of computing labels more efficiently than by using highway dimension, and leads to comparable or stronger theoretical bounds on hub set size.
\end{abstract}


\section{Introduction} 

The task of efficiently processing shortest path queries to a graph has been studied in a plethora of settings. One interesting observation is that for many real-world graphs of small degree in a geometric or geographical setting, such as road networks, it is possible to design compact data structures and schemes for efficiently answering shortest path queries. The general principle of operation of this approach consists in detecting and storing subsets of so-called transit nodes, which appear on shortest paths between many node pairs.

In an attempt to explain the efficiency of variants of the transit node routing (TNR) algorithm~\cite{tnr,tnr2}, Abraham et al.~\cite{hw10} introduced the concept of \emph{highway dimension} $h$. This parameter captures the intuition that when a map is partitioned into regions, all significantly long shortest paths out of each region can be hit by a small number of transit node vertices. The value of $h$ is presumed to be a small constant e.g.\ for road networks. However, the definition of highway dimension relies on the notion of a hitting set of shortest path sets within network neighborhoods, and hence, e.g., exact computation of the parameter is known to be NP-hard even for unweighted networks~\cite{hwhard}. This motivates us to look at other measures, which are both more locally defined and computationally tractable, while capturing essentially the same (or more) characteristics of the network's amenability to shortest path queries.

Looking more precisely at the TNR algorithm, one observes that it is built around the idea that for every source node, the set of transit nodes which are the first to be encountered when going a long way from the source, is small. This is a weaker assumption than the existence of a small hitting set for the set of shortest paths in a given network neighborhood, since different source nodes could use different transit nodes resulting in an overall large number of transit nodes around a given region. This source-centered approach leads us to the definition of \emph{skeleton dimension} $\sk$, to which we devote the remainder of this paper. Informally, the skeleton dimension is the maximum, taken over all nodes $u$ of the graph and all radii $r>0$, of the number of distinct nodes at distance $r$ from $u$ in the set of all shortest paths originating at $u$ and having length at least $3r/2$.\footnote{One may also define skeleton dimension with a different choice of constants, considering the set of shortest paths having length at least $\alpha r$, where $\alpha \in (1,2)$ is an absolute constant. The choice of $\alpha = 3/2$ is subsequently necessary only for establishing relations with highway dimension.} In transit node parlance, it states that the paths from $u$ that extend by $r/2$ at least outside the disk of radius $r$ pass through at most $\sk$ transit nodes at the disk border. This property ensures that each shortest-path spanning tree is built around a core skeleton with at most $\sk$ branches at a given distance range while the rest of the branches are relatively short. Bounding tree skeletons turns out to encompass a larger class of constant-degree graphs than the shortest path cover approach used in the definition of highway dimension, while still ensuring the existence of efficient labeling schemes.

Motivated by applications in distributed algorithms and distributed data representation, we will display the link between small skeleton dimension of a graph and efficient processing of shortest path queries using the framework of \emph{distance labeling}. Distance labeling schemes, popularized by Gavoille et al.~\cite{Gavoille:2004:DLG:1036161.1036165}, are among the most fundamental distributed data structures for graph data. Within distance labeling, we work with the most basic framework of transit-node based schemes, namely so-called \emph{hub labelings}, cf.~\cite{Abraham:2012:HHL:2404160.2404164} (this framework was first described in~\cite{Cohen:2003:RDQ:942270.944300} under the name of 2-hop covers, and is also referred to as landmark labelings~\cite{Abraham11onapproximate}). In this setting, each node $u\in U$ stores the set of its distances to some subset $S(u) \subseteq V$ of other nodes of the graph. Then, the computed distance value $d'(u,v)$ for a queried pair of nodes $u, v\in V$ is returned as:
\begin{equation}\label{eq:distance}
d'(u,v) := \min_{w \in S(u)\cap S(v)} \cost(u,w) + \cost(w,v),
\end{equation}
where $\cost$ denotes the shortest path distance function between a pair of nodes. The computed distance between all pairs of nodes $u$ and $v$ is exact if set $S(u)\cap S(v)$ contains at least one node on some shortest $u-v$ path. This property of the family of sets $(S(u) : u\in V)$ is known as \emph{shortest path cover}. The hub-based method of distance computation is in practice effective for two reasons. First of all, for transportation-type networks it is possible to show bounds on the sizes of sets $S$, which follow from the network structure. Notably, considering networks of bounded highway dimension $h$, Abraham et al.~\cite{hw10} show that an appropriate cover of all shortest paths in the graph can be achieved using sets $S$ of size $\widetilde O(h)$, where the $\widetilde O$-notation conceals logarithmic factors in the studied graph parameters.

Moreover, the order in which elements of sets $S(u)$ and $S(v)$ is browsed when performing the minimum operation is relevant, and in some schemes, the operation can be interrupted once it is certain that the minimum has been found, before probing all elements of the set. This is the principle of numerous heuristics for the exact shortest-path problem, such as contraction hierarchies and algorithms with arc flags~\cite{Kohler06fastpoint-to-point,Bauer:2010:SFR:1498698.1537599}.


\subsection{Results and Organization of the Paper}

In Section~\ref{sec:hd}, we formally define skeleton dimension $\sk$, and show that in the so-called continuous representation of the graph, the skeleton dimension is at most highway dimension, i.e. it satisfies the bound $\sk \leq h$. 
In all cases, $\sk = O(h)$ for graphs of bounded maximum degree. On the other hand, we show that skeleton dimension provides a better explanation for small hub set size in Manhattan-type networks than highway dimension. In particular, we provide a natural example of a weighted grid with $\sk = O(\log n)$ and $h = \Omega(\sqrt n)$.

In Section~\ref{sec:hub_labeling}, we show how to construct efficient hub labelings for networks of small skeleton dimension. The hub set sizes we obtain for a graph of weighted diameter $\diam$ are bounded by $O(\sk \log \diam )$ on average and $O(\sk \log \log \sk \log \diam )$ in the worst case (cf.~Corollaries~\ref{cor:klabelingaverage} and~\ref{cor:klabelingworstcase}, respectively), as compared to previous best bounds of $O(h \log h \log \diam )$ for labels computable in polynomial time based on highway dimension.

Our labeling technique, based on picking hubs through a random selection process on a subtree of the shortest-path tree, allows each node to compute its hub set independently in almost-linear time, and appears to be of independent interest. In particular, as an extension of our technique, we provide in Section~\ref{sec:dpreserving} improved bounds on label size (in general unweighted graphs) for the so-called \emph{$\DD$-preserving} distance labeling problem, in which the considered distance queries are restricted to nodes at distance at least $\DD$ from each other. The hub sets constructed using the hub-based method have average size $O(n/\DD)$. Their worst-case size is also bounded by $O(n/\DD)$ up to some threshold $\DD = \widetilde O(\sqrt n)$, and bounded by $O(\log \DD + (n / \DD) \log \log \DD)$ in general (Theorem~\ref{thm:dpreserving}). This improves upon previous $\DD$-preserving schemes, including the previously best result from~\cite{Sublinear}, where hub sets of worst-case size $O((n/\DD) \log \DD)$ are constructed by a more direct application of the probabilistic method to sets of randomly sampled vertices.

Finally, in Sections~\ref{sec:compute},~\ref{sec:generalizations},~and~\ref{sec:conclusion} we provide some concluding remarks on the computability of the proposed parameter of skeleton dimension, as well as its possible generalizations and applications.

\subsection{Other Related Work}

\paragraph{Distance Labelings.}
The distance labeling problem in undirected graphs was first investigated by Graham and Pollak~\cite{pollak}, who provided the first labeling scheme with labels of size $O(n)$. The decoding time for labels of size $O(n)$ was subsequently improved to $O(\log \log n)$ by Gavoille et al.~\cite{Gavoille:2004:DLG:1036161.1036165} and to $O(\log^* n)$ by Weimann and Peleg~\cite{WP11}. Finally, Alstrup et al.~\cite{DBLP:conf/soda/AlstrupGHP16} present a scheme for general graphs with decoding in $O(1)$ time using labels of size $\frac{\log 3}{2} n + o(n)$ bits. This matches up to low order terms the space of the currently best know distance oracle with $\bigo{1}$ time and $\frac{\log 3}{2} n^2 + o(n^2)$ total space in a \emph{centralized} memory model, due to Nitto and Venturini~\cite{NV08}. For specific classes of graphs, Gavoille et al.~\cite{Gavoille:2004:DLG:1036161.1036165} described a $O(\sqrt{n}\log n)$ distance labeling for planar graphs, together with $\Omega(n^{1/3})$ lower bound for the same class of graphs. Additionally, $O(\log^2 n)$ upper bound for trees and $\Omega(\sqrt{n})$ lower bound for sparse graphs were given.

\paragraph{Distance Labeling with Hub Sets.} For a given graph $G$, the computational task of minimizing the sizes of hub sets $(S(u) : u\in V)$ for exact distance decoding is relatively well understood. A $O(\log n)$-approximation algorithm for minimizing the average size of a hub set having the sought shortest path cover property was presented in Cohen et al.~\cite{Cohen:2003:RDQ:942270.944300}, whereas a $O(\log n)$-approximation for minimizing the largest hub set at a node was given more recently in Babenko et al.~\cite{DBLP:conf/icalp/BabenkoGGN13}. Rather surprisingly, the structural question of obtaining bounds on the size of such hub sets for specific graph classes, such as graphs of bounded degree or unweighted planar graphs, is wide open.

\paragraph{$\DD$-preserving Labeling.} The notion of $\DD$-preserving distance labeling, first introduced by Bollob\'as et al.~\cite{BCE05}, describes a labeling scheme correctly encoding every distance that is at least $\DD$. \cite{BCE05} presents such a $\DD$-preserving scheme of size $O(\frac{n}{\DD} \log^2 n)$. This was recently improved by Alstrup et al.~\cite{Sublinear} to a $\DD$-preserving scheme of size $O(\frac{n}{\DD}\log^2 \DD)$. Together with an observation that all distances smaller than $\DD$ can be stored directly, this results in a labeling scheme of size $O(\frac{n}{\symb}\log^2 \symb)$, where $\symb = \frac{\log n}{\log \frac{m+n}{n}}$. For sparse graphs, this is $o(n)$.

\paragraph{Road networks.}
Highway dimension $h$ guarantees the existence of distance labels of size $O(h\log D)$ where $D$ is the weighted diameter of the graph~\cite{hw10}. However, when restricting to polynomial time algorithms, such labels can only be approximated within a $\log n$ factor using shortest path cover algorithms~\cite{hw10} or a $\log h$ factor with a more involved procedure based on VC-dimension~\cite{hwVC}. In any case, this requires an all-pair shortest path computation. For large networks, labels can be practically computed when classical heuristics such as contraction hierarchies (CH) can be performed~\cite{hw10,hublab11,hublab12}. Low highway dimension guarantees that there exists an elimination ordering for CH such that the graph produced has bounded size~\cite{hw10}. However, it does not ensure running time faster than all pair shortest path computation.

Besides highway dimension, skeleton dimension is also related to the notion of \emph{reach} introduced in~\cite{reach} and also used in the RE algorithm~\cite{RE}. The reach of a node $v$ on a path $P$ is the minimum distance to an extremity of $P$, and the reach of $v$ is its maximum reach over all shortest paths $P$ containing $v$. Efficient algorithms are obtained by pruning nodes with small reach during Dijkstra search. Similarly, we obtain the skeleton of a tree by pruning nodes whose reach (in the tree) is less than half of their distance to the root.

\subsection{Notation and Parameters} \label{sec:definitions}

We consider a connected undirected graph $G$ and a non-negative length function
$\ell:E(G) \rightarrow {\R^+}$. Let $n$ denote the number of nodes
in $V(G)$. We let $\ell(P)$ denote the length of a path $P$ under the given length function. Given two nodes $u$ and $v$, we assume that there is a unique
shortest path $P_{uv}$ between them. This common assumption can
be made without loss of generality, as one can perturb the input to ensure
uniqueness. Given two nodes $u$ and $v$, their distance is
$d_G(u,v)=\ell(P_{uv})$. Let $\diam=\max_{u,v}d_G(u,v)$ denote the diameter of $G$.
For $u\in V(G)$ and $r>0$, the ball $B_G(u,r)$ of radius $r$ centered at $u$ is the
set of nodes $v$ with $d_G(u,v)\le r$.
In this paper, we assume that $\ell$ is non-negative and integral.
The notions presented here easily extend to non-negative real lengths, but we use integer lengths for a cleaner exposition of algorithms and theorems.

We also recall two structural parameters, which have application to networks in a geometric setting or low-dimensional topological embedding:  \emph{highway dimension} and \emph{doubling dimension}.

For $r>0$, let $P_G(r)$ denote the collection
of all shortest paths $P$ with $\frac{r}{2} < \ell(P) \le r$ in $G$. For $u\in V(G)$,
we consider the collection
$P_G(u,r)=\set{P\in P_r(G)\mid P\cap B_G(u,r)\not=\emptyset}$ of shortest
paths around $u$.
A hitting set for $P_G(u,r)$ is a set $H$ of nodes
such that any path in $P_G(u,r)$ contains a node in $H$. In \cite{hwVC}, the
\emph{highway dimension} of $G$ is defined as the smallest $h$ such that
$P_G(u,r)$ has a hitting set of size at most $h$ for all $u,r$.
(This definition is slightly less restrictive than that of
\cite{hw10} while allowing to prove similar results with improved bounds.)

The notion of highway dimension is related
to that of doubling dimension. Recall that a graph is $h$-doubling if
any ball can be covered by at most $h$ balls of half the radius.
That is, for all $u,r$, there exists $H$ with $\card{H}\le h$
such that $B_G(u,r)\subseteq \cup_{v\in H}B_G(v,\frac{r}{2})$.
It is shown in \cite{hw10} that if the geometric realization of a graph $G$
has highway dimension $h$, then $G$ is $h$-doubling. Informally, the \emph{geometric realization} $\G$  can be seen as the ``continuous'' graph where each edge
is seen as infinitely many vertices of degree two with infinitely small edges, such that for any $uv\in E(G)$ and $t\in [0,1]$, there
is a node in $\G$ at distance $t\ell(u,v)$ from $u$ on edge $uv$.
(The proof in \cite{hw10} consists in proving that any node in $B_G(u,r)$
is at distance at most $\frac{r}{2}$ from any hitting set of $P_\G(u,r)$ in $\G$ and holds also for the highway dimension definition of
\cite{hwVC}.)

\section{A Presentation of Skeleton Dimension}\label{sec:hd}

We start by providing a standalone definition of skeleton dimension based on size of cuts in shortest path trees, and then show its relation to the previously considered parameters of highway and doubling dimension.

\subsection{Definition of the Parameter}
\label{sec:skeldef}

\paragraph{Tree skeleton.}
Given a tree $T$ rooted at node $u$ with length function
$\ell:E(T)\rightarrow{\mathbb R^+}$, we treat it as directed
from root to leaves and consider the geometric realization $\T$
of this directed graph.
We define the \emph{reach} of $v\in V(\T)$ as
$\reach_{\T}(v) := \max_{x\in V(\T)}d_{\T}(v,x)$.
We then define the \emph{skeleton} $T^*$ of $T$ as the subtree
of $\T$ induced by nodes with reach at least half their distance from the root. More precisely, $T^*$ is the subtree of $\T$ induced by
$\{v\in V(\T)\mid \reach_{\T}(v) \geq \frac{1}{2} d_{\T}(u,v)\}$.

\paragraph{Width of a tree.}
The \emph{width} of a tree $T$ with root $u$
is defined as the maximum number of nodes (points)
in $\T$ at a given distance from its root. More precisely, the width of $T$
is $\width(T)=\max_{r>0}|\cut_r(\T)|$ where $\cut_r(\T)$ is the set of nodes
$v\in V(\T)$ with $d_{\T}(u,v) = r$.

\paragraph{Skeleton dimension.}
The \emph{skeleton dimension} $\sk$ of a graph $G$ is
defined as the maximum width of the
skeleton of a shortest path tree, that is $\sk=\max_{u\in V(G)}\width(T_u^*)$,
where $T_u$ denotes the shortest path tree of $u$ obtained as the
union of shortest paths from $u$ to all $v\in V(G)$.

We remark that, under the assumption of scale-invariance of the graph, different cuts of the tree skeleton have similar width, and the definition of the skeleton dimension is a meaningful measure of the structure of the tree. A smoothed (integrated) variant of skeleton dimension is also discussed further on, cf.~Eq.~\eqref{eq:isk_def}.


\subsection{Skeleton Dimension is at most (Geometric) Highway Dimension}

\begin{claim}\label{claim:skhd}
If the geometric realization $\G$ of a graph $G$ has highway dimension
$\tilde{h}$, then $G$ has skeleton dimension $k\leq \tilde{h}$.
\end{claim}

\begin{proof}
Consider a node $u$ and the skeleton $T_u^*$ of its shortest path tree $T_u$.
For $r>0$, consider the cut $\cut_r(T_u^*)$.
For $\eps>0$ sufficiently small, $\cut_r(T_u^*)$ and
$\cut_{r-\eps}(T_u^*)$ have same size. Now, for
$v\in\cut_{r-\eps}(T_u^*)$, consider a node $x$ in $T_u$
such that $d_{T_u}(v,x)=\reach_{T_u}(v)$.
The shortest path $P_{vx}$ intersects $B_G(u,r)$ and has length
$\ell(P_{vx}) = \reach_{T_u}(v) \ge r/2+\eps > r/2$.
$P_{vw}$ is thus in $P_\G(u,r)$.
For each node in $\cut_{r-\eps}(T_u^*)$, we get a similar path in $P_\G(u,r)$.
All these paths are pairwise node-disjoint as they belong to disjoint sub-branches
of $T_u$. Their number is thus upper-bounded by the size of any hitting set
of $P_\G(u,r)$. We then get $\card{\cut_r(T_u^*)}\le \tilde{h}$ for all $u,r$
and the skeleton dimension of $G$ is at most $\tilde{h}$.
\end{proof}

Note that a (discrete) graph $G$ has highway dimension $h\le\tilde{h}$,
where $\tilde{h}$ is the highway dimension of its
geometric realization $\G$.
In road networks it is expected that the continuous and the discrete versions of highway dimension coincide almost exactly, in particular due to the constant maximum degree and bounded length of edges in these graphs. In a more general setting, one can easily show $\tilde{h}\le (\Delta+1) h$ where $\Delta$ is the maximum degree of $G$, with a star being a worst-case example. (Indeed, a hitting set $H$ of $P_G(u,r)$ may miss some shortest path
$P\in P_\G(u,r)$. Making $P$ longer to have extremities in $V(G)$
transforms it into a path of $P_G(u,r)$ that is hit by $H$.
It is thus possible to hit all $P_\G(u,r)$ by adding at most one node per edge adjacent to a node in $H$.)

We remark that the extended tech-report version~\cite{MS-TR} of~\cite{hw10} introduces a modified notion of highway dimension, in a way more closely related to its geometric variant, which we can denote here as $h^*$. For this modified parameter, we have: $k \leq h^* \leq \tilde{h} \leq 2h^*$, where the first inequality follows from an analysis similar to the proof of Claim~\ref{claim:skhd}, while the latter two are shown in~\cite{MS-TR}[Section 11].


\subsection{Low Skeleton Dimension Implies Low Doubling Dimension}

It is known~\cite{hw10} that a graph having a geometric realization with highway dimension $\tilde h$ is at most $\tilde h$-doubling. However, the relation $\sk \leq \tilde h$ need not be tight, and it turns out that the link between skeleton dimension and doubling dimension holds in a slightly weaker form.

\begin{proposition}\label{pro:doubling}
If a graph $G$ has skeleton dimension $\sk$, then $G$ is $(2\sk+1)$-doubling.
\end{proposition}

\begin{proof}
We show the stronger requirement that each ball of radius $19r/9$ can be
covered by $2\sk+1$ balls of radius $r$.
For $u\in V(G)$, consider the shortest path
tree $T_u$ of $u$. For $r'>0$, consider the set $C_{r'}$ of the edges
containing a node in $\cut_{r'}(T_u^*)$ and let
$I_{r'} = \set{w\mid vw\in C_{r'} \mbox{ and } d_G(u,w) \ge r'}$ be
the (at most $\sk$) far extremities of edges cutting distance $r'$
in the skeleton of $T_u$.
Each node $v$ at distance greater than $\frac{3}{2} r'$ from $u$ is
descendant in $T_u$ of a node $x\in I_{r'}$ by skeleton definition
and is thus in $B_G(x,r)$ if $d_G(u,v)\le r'+r$. Considering $r'=2r/3$,
we obtain that any node $v$ with $r < d_G(u,v) \le \frac{5}{3} r$ is at distance at most
$r$ from a node in $I_{2r/3}$. Similarly any node $v$ with
$\frac{5}{3} r < d_G(u,v) \le \frac{19}{9} r$ is at distance at most $r$
from a node in $I_{10r/9}$. The ball $B_G(u,19r/9)$ is thus covered by balls
of radius $r$ centered at the at most $2\sk+1$ nodes in
$\set{u}\cup I_{2r/3} \cup I_{10r/9}$.
\end{proof}

\subsection{Separating Skeleton Dimension and Highway Dimension}

We now provide a family of graphs which exhibit an exponential gap between skeleton and
highway dimensions, in a setting directly inspired by Manhattan-type road networks. The idea is to consider the usual square grid and define edge lengths, which give priority to certain transit ``arteries''. In our example, paths using edges whose coordinates are multiples of high powers of 2 have slightly lower transit times.

For $L>0$, let $G_L$ denote the $2^L\times 2^L$ grid
with length function $\ell$ defined as follows. We identify a node with its
coordinates $(x,y)$ with $1\le x\le 2^L$ and $1\le y \le 2^L$.
We consider small length perturbations $p_{xy}$ for every horizontal edge
$\{(x,y),(x+1,y)\}$ and $q_{xy}$ for every vertical edge $\{(x,y),(x,y+1)\}$,
and define $Q=1+\max_{x,y}\max\{p_{xy},q_{xy}\}$. These non-negative
integers will be chosen to ensure uniqueness of shortest paths.
For $x=2^ix'$ with $0\le i \le L$ and $x'$ odd,
we define $\ell((x,y), (x+1,y))=Q((D+2)L-i)-q_{xy}$ for all $y$
where $D=2^{L+3}$.
For $y=2^jy'$ with $0\le j\le L$ and $y'$ odd,
we define $\ell((x,y), (x,y+1))=Q((D+2)L-j)-p_{xy}$ for all $x$.
A possible choice for perturbations ensuring uniqueness of shortest paths is
$p_{xy}=0$ and $q_{xy}=x$ for all $x,y$ as will be clear later on.

\begin{proposition}
For any $L>0$, grid $G_L$ has highway dimension $\Omega(\sqrt{n})$
and skeleton dimension $O(\log n)$, where $n=2^{2L}$ is the number of nodes
in $G_L$.
\end{proposition}

\begin{proof}
We first prove that the shortest paths of $G_L$ are also shortest path of the
$2^L\times 2^L$ grid $U_L$ with unit edge lengths, that is those paths that use
a minimum number of edges.
Any path $P$ with $p$ edges has length at most $pQ(D+2)L$
and at least $pQ((D+2)L - L - 1)\ge pQDL$.
Given two nodes $u$ and $v$, let $p$ denote the minimum number of edges of
a path from $u$ to $v$ and let
$q$ denote the number of edges of the shortest path $P_{uv}$ from $u$ to
$v$ in $G_L$. We then have $p\le q \le \frac{p(D+2)L}{DL} \le p+\frac{1}{2}$ since
$p\le 2^{L+1}=D/4$. This
implies $q=p$. $P_{uv}$ is thus a shortest path of the grid $U_L$.
Note that balls must then also be almost identical: for all $u,p$, we have
$B_{G_L}(u,pQ(D+2)L)=B_{U_L}(u,p)$.

\smallskip

This implies that the highway dimension of $G_L$ is
$\Omega(\sqrt{n})$ since the
ball of radius $r=Q(D+2)L\sqrt{n}$ centered at $(1,1)$ intersects at least
$\sqrt{n}$ horizontal shortest paths of length $QDL\sqrt{n} > r/2$ at least.

\smallskip

We define the $2^i\times 2^j$ rectangle $R$ at $(x,y)$ with odd x and y
as the set of nodes
with coordinates $(x',y')$ such that $2^ix\le x'\le 2^i(x+1)$
and $2^jy\le y'\le 2^j(y+1)$. Its border is the set of nodes for
which one inequality at least is indeed an equality. Other nodes are said
to be interior.  The main argument for bounding skeleton dimension
is that a shortest path from $u=(x',y')$ with $x'<2^ix$ and $y'<2^jy$
cannot traverse the interior of $R$: a shortest path passing through
an inner node of $R$ necessarily ends inside $R$. The reason
is that such a path necessarily passes through the lower left corner
at $(2^ix,2^jy)$. It is then
shorter to reach a border node by following the border rather
than using edges inside the rectangle.
Note that two possible choices of shortest paths could be possible
when going from a corner of a $2^i\times 2^i$ rectangle to the corner diagonally
opposed. However the choice of perturbing lengths by decreasing the length
of any vertical edge in position $(x,y)$ by $q_{xy}=x$ ensures that
the path through the rightmost side is preferred.

Now consider a node $u=(x,y)$ and a radius $r$ with
$2^iQDL\le r < 2^{i+1}QDL$.  Set $p=\floor{\frac{r}{QDL}}$.
According to the first part of the proof, the ball $B_{G_L}(u,r)$ is
in sandwich between the balls of radius $p$ and $p+1$ in $U_L$:
$B_{U_L}(u,p)\subseteq B_{G_L}(u,r)\subseteq B_{U_L}(u,p+1)$.
We first consider the upper
right quadrant of $B_{U_L}(u,p)$ and the border set $S_{UR}$
of nodes $v=(x+a,y+b)$ with
$a,b\ge 0$ and $a+b=p$.  We now bound the number of nodes $v\in S_{UR}$ such that
$\reach_{T_u}(v)>\frac{r}{2}$ where $T_u$ denote the shortest path tree of $u$
in $G_L$. As $\frac{r}{2}\ge 2^{i-1}QDL$, such a node $v$ cannot be interior to
a $2^{i-2}\times 2^{i-2}$ rectangle as shortest paths interior to the rectangle
have length at most $(2^{i-1}-4)Q(D+2)L$.  The number of nodes
in $S_{UR}$ whose $x$
coordinate is a multiple of $2^{i-2}$ is bounded by $8$ as $p<2^{i+1}$.
Similarly, the number of nodes whose $y$ coordinate is a multiple of $2^{i-2}$
is also bounded by $8$. Apart from these $16$ nodes, we have to consider nodes
$v=(x',y')$ where $x'$ (resp. $y'$) is less than the smallest multiple of
$2^{i-2}$ greater than $x$ (resp. $y$). Such a node cannot be interior to
a $2^j\times 2^{i-2}$ (resp. $2^{i-2}\times 2^j$) rectangle for $j\le i-3$.
For $j=i-3$, we obtain at most two nodes whose $x$ coordinate is a multiple
of $2^j$. By repeating the argument for $j=i-4..1$, we can finally bound
the number of nodes $v\in S_{UR}$ with reach greater than $r/ 2$ by
$4(i-3)+16=4(i+1)$. Nodes in the upper right quadrant that are
at distance $r$ from $u$ in $\wtilde{T_u}$
must be on edges outgoing from nodes in $S_{UR}$ and $\cut_r(T_u^*)$
has thus size at most $8(i+1)=O(\log n)$.
By symmetry, the bound holds for other quadrants
and the skeleton dimension of $G_L$ is $O(\log n)$.
\end{proof}

We remark that there exist different lengths functions on the grid for which the skeleton dimension is also as large as $\Theta(\sqrt n)$. This is the case, for example, for a grid with unit lengths of all edges except for edges intersecting its major diagonal, which is configured to be a fast transit artery (it suffices to set $\ell((x,y),(x+1,y)) = 0.5$ for all $x=y$).

\medskip

We complement this result with experimental observation in real grid like networks such as encountered in Brooklyn. We computed the skeleton dimension of the New York travel-time graph proposed in the 9th DIMACS challenge~\cite{dimacs9th} which turns out to be $k=73$. (Average skeleton tree width is 30, but a maximum width of 73 is encountered for a skeleton tree rooted in Manhattan.)  In order to estimate the highway dimension of this graph, we have implemented a heuristic for finding a large packing of paths near a given ball, that is a set of disjoint paths intersecting the ball and having length greater than half radius. We could find a packing of 172 paths in Brooklyn.
This proves that the highway dimension of this graph is 172 at least ($h\ge 172$). In comparison the skeleton tree of the center of the corresponding ball has width 48, and 42 branches are cut at radius distance (see Figure~\ref{fig:brooklyn}).

\begin{figure}
\centerline{%
\includegraphics[width=.466\textwidth]{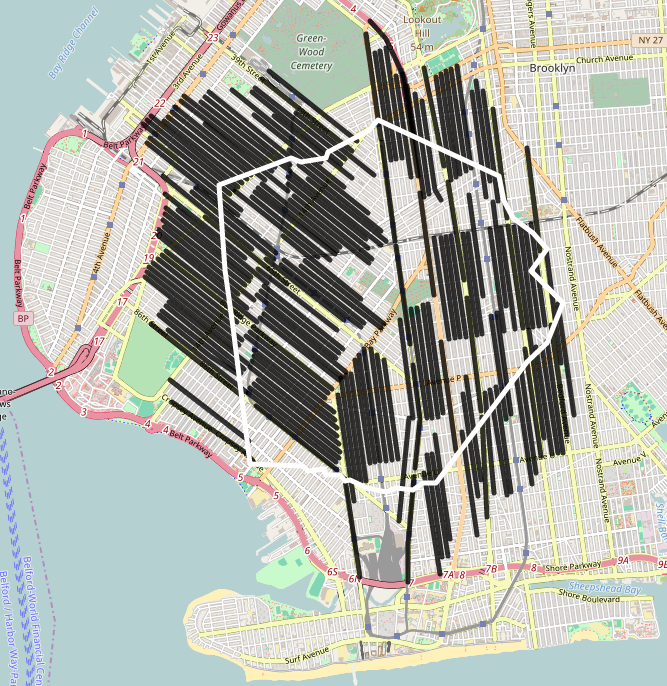}%
\hfill
\includegraphics[width=.512\textwidth]{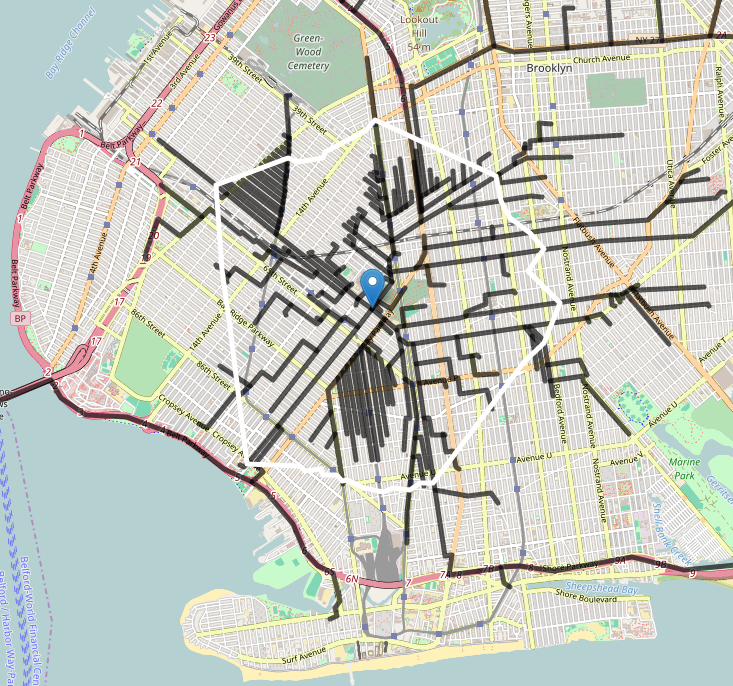}%
}
\caption{An OpenStreetMap view of Brooklyn, with a packing of 172 paths (in black) intersecting a ball of radius 720 seconds (with white border) on the left, and the skeleton tree (in black) of the center of the ball on the right.}
\label{fig:brooklyn}
\end{figure}


\section{Hub Labeling using Tree Skeletons} 
\label{sec:hub_labeling}

In this section, we assign shortest-path-intersecting hub sets to a set of (terminal) nodes $V$ of the considered network. We will assume that the length function $\ell$ on edges is integer weighted. To emulate the geometric realization of the graph, we subdivide edges into sufficiently short fragments by inserting a set of additional nodes $V^+$ into the network. For convenience of subsequent analysis, we assume that an edge $vw$, for $v, w\in V$, of integer length $\ell(vw)$ is subdivided into $12 \ell(vw)$ edges of length $1/12$ each. After this, all edges have the same length, and we subsequently treat the graph as unweighted. All the parameter definitions carry over directly from the geometric setting; for the sake of precision, we formally state the assumptions on the studied setting below.

We consider an unweighted graph $G = (V \cup V^+,E)$, with a distinguished set of terminal nodes $V$ and where all nodes from $V^+$ have degree $2$. We denote $n:=|V|$. We assume that every node $u\in V$ is associated with a fixed (unweighted) tree $T_u \subseteq G$. Throughout the section, we will denote by $P_u(v,w)$ the unique path between the pair of nodes $v$ and $w$ in tree $T_u$, and more concisely $P_u(v) := P_u(u,v)$. Where this does not lead to confusion, we will identify a path with its edge set, and we will also use the symbol $|P|$ to denote the \emph{length} of path $P$, i.e., the number of edges belonging to $P$. We write $d_u (v) := |P_u(v)|$. We require that the collection of trees $\{T_u\}_{u \in V}$ satisfies the following property: For any pair of nodes $u,v\in V$, we have $P_{u}(v) = P_{v}(u)$. We will also assume that for all $u, v, w\in V$, we have that $|P_u(v,w)|$ is an integer multiple of $12$.

We remark that if the graph $G$ was obtained by a distance-preserving subdivision of nodes of an edge-weighted graph on node set $V$ under some distance metric $\ell$, then each tree $T_u \in G$  corresponds to the shortest path tree of node $u$ under the original distance metric, and the assumption $P_{u}(v) = P_{v}(u)$ corresponds to the assumption of uniqueness of shortest paths under the original metric $\ell$.

An \emph{edge hub labeling} is an assignment of a set of edges $S(u) \subseteq E$ to each node $u \in V$, such that the following property is fulfilled: for every pair of nodes $u, v \in V$, there exists an edge $\eta \in S(u) \cap S(v)$ such that $\eta  \in P_{u}(v)$. The set $S(u)$ is known as the \emph{edge hub set} of $u$. We remark that this edge-based notion of hub sets is slightly stronger than an analogous vertex-based notion: indeed, knowing that an edge $\eta \in P_u(v)$, we also conclude that both of the endpoints of edge $\eta$ belong to $P_{u}(v)$. We choose to work with edge hub sets rather than node hub sets in this Section for compactness of arguments.

We restate in the setting of the family of trees $\{T_u\}_{u\in V}$ the notion of the \emph{skeleton} $T^*_u = T_u [V^*_u] \subseteq T_u$ as the subtree of $T_u$ induced by node set $V_u^*  = \{v \in V(T_u) : \reach_{T_u}(v) \geq \frac{1}{2} d_u(v)\}$. For $u \in V$, the width $\width(T_u^*)$ of the skeleton $T_u^*$ may be written as $\width(T_u^*) = \max_{r\in \N} |\cut_u^{*(r)}|$, where:
$$
\cut_u^{*(r)} := \{v \in V_u^* : d_u(v) = r\}.
$$
Finally, we note that the skeleton dimension of graph $G$ may be written as $\sk = \max_{u \in V} \width(T_u^*)$.

\subsection{Construction of the Hub Sets}

The edge hub sets $S(u)$, $u\in V$, are obtained by the following randomized construction. Assign to each edge $e\in E$ a real value $\rho(e) \in [0,1]$, uniformly and independently at random. We condition all subsequent considerations on the event that all values $\rho$ are distinct, $|\rho(E)| = |E|$, which holds with probability $1$.

For all $u, v\in V$, we define the \emph{central subpath} $P_u'(v) \subseteq P_u(v)$ as the subpath of $P_u(v)$ consisting  of its middle $\frac{d_u(v)}{6}$ edges; formally, $P_u'(v) := P_u(u', v')$, where $u', v' \in P_u(v)$ are nodes given by: $d_u(u') = \frac{5}{12} d_u(v)$ and $d_u(v') = \frac{7}{12} d_u(v)$.
Next, for all $u, v \in V$, $v\neq u$, we define the hub edge $\eta_u(v) \in P_u(v)$ as the edge with minimum value of $\rho$ on the central subpath between $u$ and $v$:
$$
\eta_u(v) = \arg\min_{e \in P'_u(v)} \rho(e).
$$
Finally, for each node $u\in V$, we adopt the natural definition of edge hub set $S(u)$ as the set of all edge hubs of node $u$ on paths to all other nodes:
$$
S(u) := \left\{\eta_{u}(v) : v \in V, v\neq u \right\}.
$$

\paragraph{Proof of Correctness.} Taking into account that for all $u, v\in V$, $P_u(v) = P_v(u)$, we observe that by symmetry of the central subpath with respect to its two endpoints, we also have $P_u'(v) = P_v'(u)$. It follows directly that $\eta_u(v) = \eta_v(u)$. Hence, we have $\eta_u(v) \in S(u) \cap S(v)$, and also $\eta_u(v) \in P_u(v)$, which completes the proof of correctness of the edge hub labeling.

We devote the rest of this Section to bounding the size of hub sets $S(u)$.

\subsection{Bounding Average Hub Set Size}

In subsequent considerations we will fix a node $u \in V$, and restrict considerations to the tree $T_u$. We will assume that tree $T_u$ is oriented from its root $u$ towards its leaves, and we will call a path $P\subseteq T_u$ a \emph{descending path} in $T_u$ if one of its endpoints is a descendant of the other in $T_u$. In particular, every path $P_u(v)$ is a descending path. For an edge $e\in T_u$, we will denote by $e^+$ and $e^-$ the two endpoints of $e$, with $e^-$ being the one further away from the root ($d_u(e^-) = d_u(e^+) + 1$). Likewise, for a descending path $P$, we denote by $P^+$ and $P^-$ its two extremal vertices, closest and furthest from the root $u$, respectively. We also denote by $d_u(e) := d_u(e^-)$ the distance of edge $e$ from the root.

In order to bound the expected size of the hub set $S(u)$, we will observe that elements of $S(u)$ necessarily belong to the skeleton $T^*_u$ and satisfy certain minimality constraints with respect to descending paths of sufficiently large length, contained entirely within the skeleton $T^*_u$.

\begin{lemma}\label{lem:hubset_properties}
Let $\eta \in S(u)$, for some $u \in V$. Then, the following claims hold:
\begin{enumerate}
\item[(1)]  $\eta \in E(T^*_u)$.
\item[(2)]  $\reach_{T^*_u} (\eta^-) \geq \frac{1}{7} d_u (\eta)$.
\item[(3)]  There exists a descending path $P \subseteq T^*_u$, such that $\eta = \arg\min_{e \in P} \rho(e)$ and $|P| \geq \frac{2}{7} d_u (\eta)$.
\end{enumerate}
Furthermore, the following claims hold for any edge $\eta \in E$ satisfying Claims (1) and (3):
\begin{enumerate}
\item[(4)]  There exists a descending path $P \subseteq T^*_u$, such that $\eta = \arg\min_{e \in P} \rho(e)$, $\eta$ is one of the two extremal edges of $P$ (i.e., $P^+=\eta^+$ or $P^-=\eta^-$), and $|P| = \left \lceil \frac{1}{7} d_u (\eta) \right \rceil$.
\item[(5)]  There exists a descending path $P_u(x,y)$, for some $y \in V^*_u$ and $x\in V^*_u$ satisfying $d_u(x) = \left \lfloor\frac{7}{8}d_u(y)\right\rfloor$, such that $\eta = \arg\min_{e \in P_u(x,y)} \rho(e)$ and $\eta$ is one of the two extremal edges of $P_u(x,y)$.
\end{enumerate}
\end{lemma}
\begin{proof}
Select $\eta \in S(u)$ arbitrarily. Let $v \in V$ be any node such that $\eta = \eta_u(v)$. We recall that for the descending path $P_u(u',v') \subseteq P_u(v)$, we have $\eta = \arg\min_{e \in P_u(u',v')} \rho(e)$, where $d_u(u') = \frac{5}{12} d_u(v)$ and $d_u(v') = \frac{7}{12} d_u(v)$. Let $v'' \in  P_u(v)$ be a node such that $d_u(v'') = \frac{2}{3} d_u(v)$ (we recall that $12 | d_u(v)$ by assumption). By the definition of skeleton $T^*_u$, we have $v'' \in V^*_u$, and clearly $P_u(v'') \subseteq T^*_u$. We note that $\eta \in P_u(v'')$ and moreover:
$$
\reach_{T^*_u} (\eta^-) \geq |P_u(\eta^-,v'')| \geq |P_u(v',v'')| = \frac{2}{3} d_u(v) - \frac{7}{12} d_u(v) = \frac{1}{12} d_u (v) = \frac{1}{7} d_u (v') \geq \frac{1}{7} d_u (\eta),
$$
hence Claims (1) and (2) follow. To show Claim (3), we put $P = P_u(u',v')$ and observe that $\eta = \arg\min_{e \in P} \rho(e)$ by definition, and: $$
|P| = |P_u(u',v')| = \frac{7}{12} d_u(v) - \frac{5}{12} d_u(v) = \frac{1}{6} d_u(v) = \frac{2}{7} d_u(v') \geq \frac{2}{7} d_u (\eta).
$$
Next, to show Claim (4), we observe that by Claim (3), $\eta = \arg \min_{e \in P_u(u', \eta^-)} \rho(e)$ and $\eta = \arg \min_{e \in P_u(\eta^+, v')} \rho(e)$. Moreover, since $|P_u(u',v')| \geq \frac{2}{7} d_u (\eta)$, we have $|P_u(u',\eta^-)| + |P_u(\eta^+,v')| \geq 1+ \frac{2}{7} d_u (\eta) \geq 2 \left \lceil \frac{1}{7} d_u (\eta) \right \rceil$, and so  $|P_u(u',\eta^-)| \geq \left \lceil \frac{1}{7} d_u (\eta) \right \rceil$ or $|P_u(\eta^+,v')| \geq \left \lceil \frac{1}{7} d_u (\eta) \right \rceil$. Claim (4) thus follows for an appropriate choice of descending path $P \subseteq |P_u(u',\eta^-)|$ or $P \subseteq |P_u(\eta^+,v')|$, respectively, with $|P| = \left \lceil \frac{1}{7} d_u (\eta) \right \rceil$.

Finally, to show Claim (5), we consider separately the two cases from the proof of Claim (4).

If $P \subseteq P_u(u',\eta^-)$, then we set $y = \eta^-$, and choose $x \in P_u(\eta^-)$ so that $d_u(x) = \lfloor \frac{7}{8} d_u(\eta^-) \rfloor$. We have $d_u(x,y) = \lceil \frac{1}{8} d_u(\eta^-) \rceil \leq \lceil \frac{1}{7} d_u(\eta^-) \rceil = |P|$, hence $P_u(x,y) \subseteq P$, $\eta = \arg \min_{e\in P_u{(x,y)}} \rho(e)$, and the claim follows.

If $P \subseteq P_u(\eta^+,v')$, then we set $x = \eta^+$, and choose $y \subseteq P_u(x,v')$ so that $d_u(x) = \lfloor \frac{7}{8} d_u(y) \rfloor$. Such a choice of $y$ is always possible since, when moving with $y$ along the path $P$, the value of $\lfloor \frac{7}{8} d_u(y) \rfloor$ increases by at most $1$ in every step; moreover, for the lower end node $v^*$ of path $P$ we have $d_u(v^*) > d_u(x) + \frac{1}{7} d_u (x)$, and so $d_u(x) \leq \lfloor \frac{7}{8} d_u(v^*) \rfloor$. We again obtain $\eta = \arg \min_{e\in P_u{(x,y)}} \rho(e)$, and the claim follows.
\end{proof}

We remark that the remainder of our analysis is valid for any construction of hub sets $S$ which satisfies Claims (1), (2), and (3) of Lemma~\ref{lem:hubset_properties}.\footnote{One may, in particular consider an alternative construction of a hub set $S^+(u)$, defined as the set of all edges $\eta \in E(T_u)$ which satisfy Claims (1), (2), and (3) of the Lemma. All of our bounds on hub set size also hold in the case of $\{S^+(u)\}_{u\in V}$. 
The definition results in larger labels in practice: we always have $S(u) \subseteq S^+(u)$ (correctness results from this observation). On the other hand, hub sets $S^+(u)$ may sometimes be constructed more efficiently than $S(u)$: the definition of $S(u)$ requires a scan of the entire tree $T(u)$, whereas hub set $S^+(u)$ may be constructed based only on the smaller skeleton $T^*(u)$.}

To bound the average hub set size precisely, we introduce for each node $u$ a parameter called \emph{integrated skeleton dimension} $\isk(u)$, defined through a sum of inverse distances to $u$ over nodes of its tree skeleton:
\begin{equation}
\label{eq:isk_def}
\isk(u) :=\sum_{v \in V_u^*} \frac{1}{d_u(v)} = \sum_{r \in \N} \frac{|\cut_u^{*(r)}|}{r},
\end{equation}
where the equivalence of the two definitions follows directly from the definition of cuts, $\cut_u^{*(r)} = \{v \in V^*(u) : d_u(v) = r\}$.

Taking into account that $|\cut_u^{*(r)}| \leq \width(T^*_u)$, we have $\isk(u) = O(\width(T^*_u) \log \diam(T_u^*))$, and even more roughly, we have for all $u \in V$:
\begin{equation}\label{eq:isk_sk}
\isk(u)= O(\sk \log \aspect),
\end{equation}
where we recall that $\aspect = \max_{u \in V}\diam(T_u)$.

\begin{lemma}\label{lem:expected_hub_precise}
The expected hub set size of a node $u \in V$ satisfies the bound:
$$
\E |S(u)| \leq 16 \isk(u).
$$
\end{lemma}
\begin{proof}
Fix $u\in V$ arbitrarily. For $y \in V^*_u$, we define $x_y$ as the unique node on the path $P_u(y)$ such that $d_u(x_y) = \left \lfloor\frac{7}{8}d_u(y)\right\rfloor$. We define random variable $Q_u(y) \in \{0,1\}$ as the number of extreme edges $e$ of the path $P_u(x_y, y)$ (i.e., $e^+ = x_y$ or $e^- = y$), which satisfy $e = \arg \min \rho(P_u(x_y, y))$. We have:
$$
\Pr [Q_u(y) = 1] = \frac{2}{|P_u(x_y, y)|} \leq \frac{2}{d_u(y) - \frac{7}{8}d_u(y)} = \frac{16}{d_u(y)}.
$$
By Claim (5) of Lemma~\ref{lem:hubset_properties}, it follows that by summing random variables $Q_u(y)$ exhaustively over all vertices $y$ we count each element $\eta \in S(u)$ at least once, hence:
$$
|S(u)| \leq \sum_{y \in V^*_u} Q_u(y).
$$
By linearity of expectation, it follows that:
$$
\E |S(u)| \leq \sum_{y \in V^*_u} \E Q_u(y) = \sum_{y \in V^*_u} \Pr [Q_u(y) = 1] \leq \sum_{y \in V^*_u} \frac{16}{d_u(y)}.
\vspace*{-10mm}
$$
\vspace*{3mm}
\end{proof}

A direct application of Markov's inequality to the bound from Lemma~\ref{lem:expected_hub_precise}, combined with Eq.~\eqref{eq:isk_sk}, gives the following Corollary.

\begin{corollary}\label{cor:klabelingaverage}
The average hub set size satisfies $\frac{1}{n}\sum_{u \in V}|S(u)| = O(\frac{1}{n}\sum_{u \in V} \isk(u)) \leq O(\sk \log \aspect)$, with probability at least $1/2$ w.r.t.\ choice of random values $\rho$.\qed
\end{corollary}
Obtaining concentration bounds on the maximal size of a hub set, $\max_{u\in V} |S(u)|$, requires some more care, and we proceed with the analysis in the following subsection.

\subsection{Concentration Bounds on Hub Set Size}

For fixed $u\in V$, we consider the size of the hub set of $u$ given by the random variable:
\begin{equation}\label{eq:size}
|S(u)| = \sum_{\eta \in E(T^*_u)} X(\eta),
\end{equation}
where $X(\eta) \in \{0,1\}$ is the indicator variable for the event ``$\eta\in S(u)$''. The random variables $\{X(\eta)\}_{\eta \in E(T^*_u)}$ need not, in general, be independent or negatively correlated. In the subsequent analysis, for fixed $\eta$, we make use of Claim (4) of Lemma~\ref{lem:hubset_properties} to bound random variable $X(\eta)$. By the Claim of the Lemma, we can decompose $X(\eta)$ into the contributions from descending paths located towards the root and away from the root with respect to $\eta$:
$$
X(\eta) \leq X^+(\eta) + X^-(\eta),
$$
where we define:
\begin{itemize}
\item $X^+(\eta) \in \{0,1\}$ is the indicator variable for the event: ``for the unique descending subpath $P \subseteq P_u(\eta^-)$ of length $|P| = \left \lceil \frac{1}{7} d_u (\eta) \right \rceil$ ending in edge $\eta$ (i.e., $P^- = \eta^-$), it holds that $\eta = \arg\min \rho(P)$'',
\item $X^-(\eta) \in \{0,1\}$ is the indicator variable for the event: ``there exists a descending path $P \subseteq T^*_u$ of length $|P| = \left \lceil \frac{1}{7} d_u (\eta) \right \rceil$ starting in edge $\eta$ (i.e., $P^+ = \eta^+$), such that $\eta = \arg\min \rho(P)$''.
\end{itemize}
Moreover, by Claim (2) of Lemma~\ref{lem:hubset_properties}, we may have $X(\eta)\neq 0$ only for those edges $\eta$ for which $\reach_{T^*_u} (\eta^-) \geq \frac{1}{7} d_u (\eta)$. We denote $V^{**}_u = \{ v \in V^*_u : \reach_{T^*_u} (v) \geq \frac{1}{7} d_u (v)\}$ and $T^{**}_u = T_u [V^{**}_u]$.
We can now rewrite Eq.~\eqref{eq:size} as:
\begin{equation}
\label{eq:StoX}
|S(u)| \leq \sum_{\eta \in E(T^{**}_u)} X^+(\eta) + \sum_{\eta \in E(T^{**}_u)} X^-(\eta),
\end{equation}
and proceed to bound both of these sums separately. In order to be able to manipulate the sums more conveniently, we first introduce a partition of the tree according to geometrically increasing scales of distance.

\paragraph{Partition of $T^*_u$ into Layers.} We consider a sequence of increasing integer radii $(r^{[i]})_{i\in \N}$, given as $r^{[0]} = 0$ and $r^{[i]} = \left\lceil 15 \left( \frac{16}{15}\right)^i\right\rceil,$ for $i\geq 1$. The last non-empty layer corresponds to index $i_{max} < \log_{16/15} \diam < 16 \ln \diam$.
Cutting the edge set of tree $T^*_u$ at vertices located at distances $\{r^{[i]}\}_{i\in \N}$ from the root $u$ yields the following partition into layers:
$$
E(T^*_u) = \bigcup_{i \in \N} L^{*[i]}, \quad \text{where:}\quad L^{*[i]}:= \{e\in E(T^*_u) : r^{[i]} < d_u(e) \leq r^{[i+1]}\}.
$$
We further denote the subset of each layer restricted to edges from $T^{**}_u$ as $L^{**[i]} := L^{*[i]} \cap E(T^{**}_u)$, $i\in \N$.

\begin{lemma}\label{lem:path_decomposition}
For all $i\geq 1$, edge set $L^{**[i]}$ admits a partition into paths, $L^{**[i]} = \bigcup_{j=1}^{l_i} P^{[i,j]}$, such that $l_i < 2\min_{r \in [r^{[i+1]},r^{[i+2]}]} |\cut_u^{*(r)}| \leq 2\sk$, each $P^{[i,j]}$ is a descending path, and all internal nodes of all paths $P^{[i,j]}$ have degree exactly $2$ in $T^{**}_u$.
\end{lemma}
\begin{proof}
Define partition $L^{**[i]} = \bigcup_{j=1}^{l_i} P^{[i,j]}$ of the edge set of the considered layer so that each path $P^{[i,j]}$, $1\leq j \leq l_i$, is a maximal descending path whose internal nodes all have degree exactly $2$ in $T^{**}_u$. Let $F = T_u(L^{**[i]})$ be the oriented sub-forest of $T_u$ induced by the edges of $L^{**[i]}$. Let $l$ be the number of leaves of $F$ and $c$ be the number of its connected components. An elementary relation between the number of leaves and the number of nodes of degree more than $2$ in a forest gives $l_i \leq 2 l - c < 2l$. Moreover, by the definition of  $T^{**}_u$, we have that for each $\eta \in L^{**[i]}$, we have $\reach_{T^*_u} (\eta^-) \geq \frac{1}{7} d_u (\eta^-)$. It follows that each of the $l$ paths $P^{[i,j]}$, such that $P^{[i,j]-}$ is a leaf of $F$, can be extended along a descending path in $T^*_u$ by a distance of $\frac{1}{7} d_u (P^{[i,j]-}) \geq \frac{1}{7} r^{[i]} \geq r^{[i+2]} - r^{[i]}-1$. It follows that each of the $l$ leaves of $F$ can be extended along a (independent) descending path until radius $r^{[i+2]}$ inclusive. Thus, $l < \min_{r \in [r^{[i+1]},r^{[i+2]}]} |\cut_u^{*(r)}| \leq \sk$, which completes the proof.
\end{proof}

\paragraph{Bounding the Sum of $X^+(\eta)$.}

Denote by $\mathcal{P}^{[i]}$ the set of descending paths $P$ in $T_u$ which stretch precisely between the endpoints of the $i$-th layer: $d_u(P^+) = r^{[i]}$, $d_u(P^-) = r^{[i+1]}$. For a fixed path $P^{[i,j]} \subseteq L^{**[i]}$, $i\geq 1$, we denote by $Q^{[i,j]}$ the unique path in $\mathcal{P}^{[i-1]}$ which extends to $P^{[i,j]}$, i.e., $Q^{[i,j]} \in \mathcal{P}^{[i-1]}$ and $Q^{[i,j]} \subseteq P_u(P^{[i,j]+})$.

Consider now an arbitrary edge $\eta$ which does not belong to layers $0$ or $1$ of the tree partition, $\eta \in E(T^{**}_u) \setminus (L^{**[0]} \cup L^{**[1]})$. Taking into account the above decomposition of set $T^{**}_u$ into layers, and of layers into paths, there exists a unique path $P^{[i,j]}$, such that $\eta\in P^{[i,j]}$.
Then, we observe that for the event $X^+(\eta)=1$ to hold, it is necessary that two conditions are jointly fulfilled: $\eta$ must satisfy the \emph{prefix minimum} condition on the path $P^{[i,j]}$:
\begin{equation}\label{eq:prefix_min}
\eta = \arg\min_{e \in P^{[i,j]} \cap P_u(\eta^-)} \rho(e),
\end{equation}
and moreover, we must have $\rho(\eta) < \min \rho(Q^{[i,j]})$. Indeed, considering the definition of $X^+(\eta)$, the unique descending subpath $P \subseteq P_u(\eta^-)$ of length $|P| = \left \lceil \frac{1}{7} d_u (\eta) \right \rceil$ which ends with edge $\eta$ has its other endpoint in $L^{*[i-2]}$. We have $\eta = \arg\min \rho(P)$, and path $P$ includes as subpaths both the entire prefix $P^{[i,j]} \cap P_u(\eta^+)$, and the path $Q^{[i,j]}$.

We denote by $M^{+[i,j]} \subseteq P^{[i,j]}$ the set of all edges $\eta \in P^{[i,j]}$ satisfying $\rho(\eta) < \min \rho(Q^{[i,j]})$. We further denote by $\eta^{+[i,j,k]}$ the $k$-th edge in $M^{+[i,j]}$, when ordering edges of $M^{+[i,j]}$ by increasing distance from the root $u$, $1\leq k \leq |M^{+[i,j]}|$. Finally, we denote by $Z^{+[i,j,k]} \in \{0,1\}$ the indicator random variable for the event that ``edge $\eta^{+[i,j,k]}$ satisfies the prefix minimum condition~\eqref{eq:prefix_min} on path $P^{[i,j]}$''. It follows that:
$$
\sum_{\eta \in E(T^{**}_u)} X^+(\eta) \leq \sum_{\eta \in L^{**[0]} \cup L^{**[1]}} X^+(\eta)  + \sum_{\substack{i = 2,3,4,\ldots\\ 1\leq j \leq l_i}} \sum_{k=1}^{|M^{+[i,j]}|} Z^{+[i,j,k]},
$$
where we note that the ranges of sum indices $i,j$ do not depend on the random choice of $\rho$ in our setting. We further rewrite the above expression, roughly bounding the first sum by cardinality, and splitting the second double sum according to even and odd values of $i$:
\begin{equation}
\label{eq:Xplus}
\sum_{\eta \in E(T^{**}_u)} X^+(\eta) \leq \sum_{r=0}^{19} |\cut_u^{*(r)}| +  \underbrace{\sum_{\substack{i = 2,4,6,\ldots\\ 1\leq j \leq l_i}} \sum_{k=1}^{|M^{+[i,j]}|} Z^{+[i,j,k]}}_{:=A_{even}^+} + \underbrace{\sum_{\substack{i = 3,5,7,\ldots\\ 1\leq j \leq l_i}} \sum_{k=1}^{|M^{+[i,j]}|} Z^{+[i,j,k]}}_{{:=A_{odd}^+}}.
\end{equation}
We subsequently consider only bounds on the summed expression $A_{even}^+$ for $2|i$ (bounds on the expression $A_{odd}^+$ follow by identical arguments).

We observe that for fixed $i$, $i=2a$ for some $a\in N^+$, the random variables $|M^{+[i,j]}|$ depend only on the choice of random values $\rho(e)$ for $e \in L^{*[i-1]}$. Now, conditioning on a choice $\rho_{odd}$ of random values $\rho(e)$ for $e \in L^{*[2a-1]}$, for all $a\in \N^+$, we observe that $\{Z^{+[i,j,k]}\}_{2|i, j, 1\leq k\leq |M^{+[i,j]}|}$ is a set of independent random variables, with:
$$
\Pr[Z^{+[i,j,k]} =1 | \rho_{odd}] = 1/k.
$$
The above probability and independence follows directly from a well-known characterization of the probability that the $k$-th element of a uniformly random permutation (ordering) is its prefix minimum.

We have:
\begin{equation}\label{eq:N}
\E [A_{even}^+ | \rho_{odd}] = \sum_{\substack{i = 2,4,6,\ldots\\ 1\leq j \leq l_i}} \sum_{k=1}^{|M^{+[i,j]}|} \frac1k <
\sum_{\substack{i = 2,4,6,\ldots\\ 1\leq j \leq l_i}} (1+\ln^\circ |M^{+[i,j]}|)\quad =: N_{even}^+,
\end{equation}
where $\ln^\circ x = \ln x$ for $x>0$ and $\ln^\circ 0 = 0$. By an application of a simple Chernoff bound for the sum of variables $\{Z^{+[i,j,k]}\}_{i,j,k}$, we have:
\begin{equation}\label{eq:Abound}
\Pr [A_{even}^+ < N_{even}^+ + 3c\ln n] >1- n^{-c}, \text{\quad for $c>1$}.
\end{equation}

It now remains to provide bounds on the concentration of random variable $N_{even}^+$ in its upper tail. If our only goal is to bound the hub set size as $O(\sk \log \diam \log \log n + \log n)$, then obtaining such bounds becomes a relatively straightforward exercise in Chernoff bounds over individual paths $P^{[i,j]}$. In this work, we go for a more pedestrian approach through a type of balls-into-bins process, optimizing bounds over larger path sets, which will eventually give us slightly tighter bounds, including a bound of $O(\sk \log \diam \log \log \sk)$.

Denote $M^+_i := \sum_{j=1}^{l_i}|M^{+[i,j]}|$, for $i\geq 2$. Then, we have the following bound.
\begin{lemma}\label{lem:layer_geometric}
Fix $i \geq 2$. Then, for any $c\geq 1$:
$$
\Pr \left[M^+_i > 12 c l_i \right] < 2^{-c}.
$$
\end{lemma}
\begin{proof}
For fixed $i$, we consider the edge set $F := \bigcup_{j=1}^{l_i} P^{[i,j]} \cup \bigcup_{j=1}^{l_i}Q^{[i,j]}$ (forming a subforest of $T^*_u$, contained entirely within layers $L^{*[i-1]}$ and $L^{*[i]}$). See Figure~\ref{fig:proof} for an illustration.

\begin{figure}\label{fig:proof}
\centering\hspace*{1.5cm}\includegraphics[width=15cm]{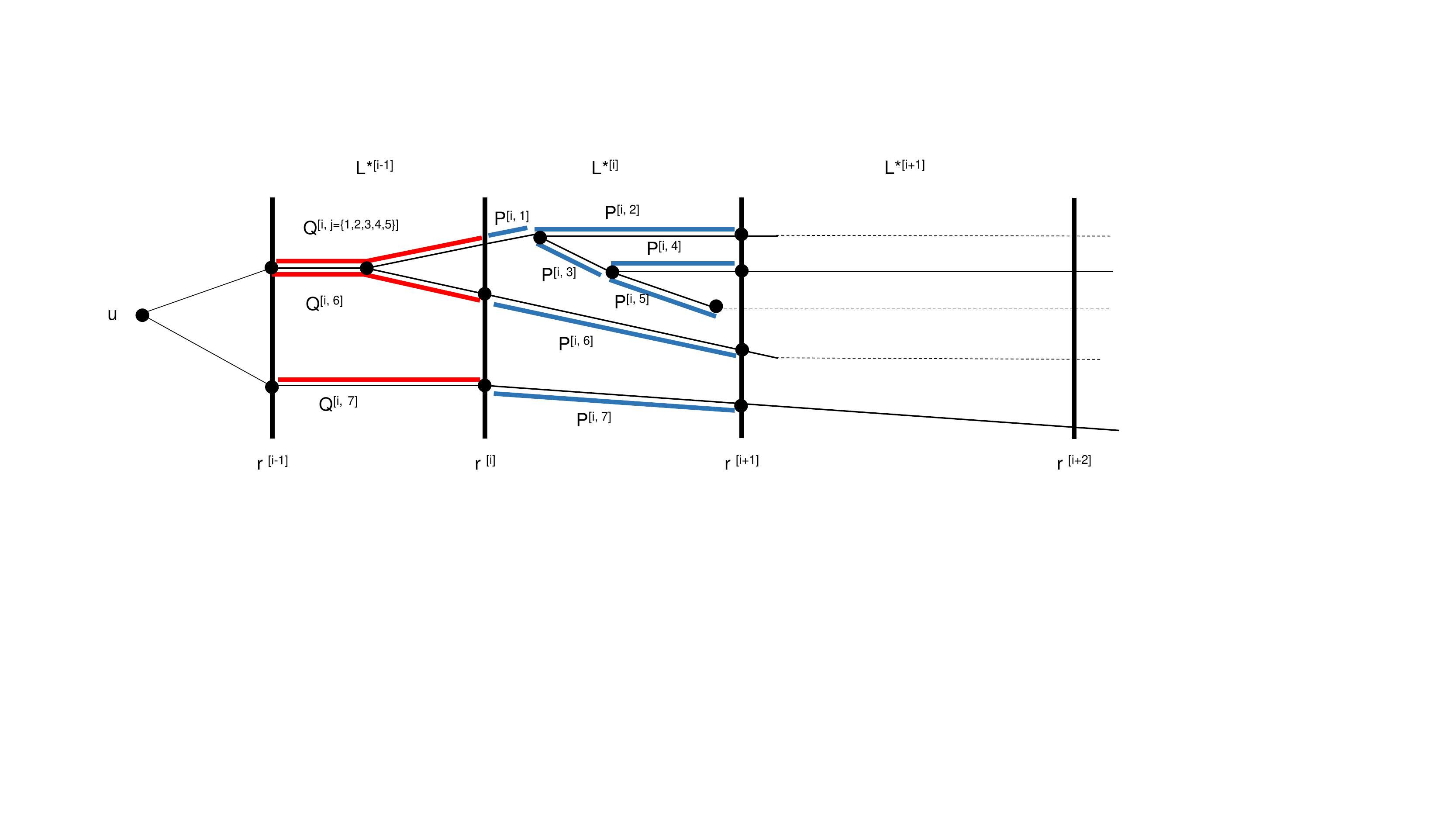}
\vspace*{-3.8cm}
\caption{Illustration of paths $P^{[i,j]}$ and $Q^{[i,j]}$. Edges from $T_u^{**}$ are marked with solid lines, remaining edges from $T_u^*$ are marked with dashed lines.}
\end{figure}
 Choose arbitrarily the set $I$ of (necessarily distinct) values of $\rho$ which appear within $F$, $I = \rho(F)$. Let  $I = \{i_1,i_2,\ldots, i_{|F|}$\}, with $i_1 \leq i_2 \leq \ldots \leq i_{|F|}$.

We couple the sampling of values $\rho$ on $F$ with the following two-stage process. First, we fix set $I$. Then, given a choice of $I$, we select a uniformly random permutation to perform the assignment of values from $I$ to edges in $F$ ($|I|=|F|$). The latter permutation is defined iteratively, by assigning to successive values $i_t$, $1\leq t \leq |F|$, an as yet unoccupied edge (site) from $F$. The value of $|M^{+[i,j]}|$ is given as the number of elements of $I$ which are placed in sites from $P^{[i,j]}$ before the smallest index $t_j\in \N$, such that $i_{t_j} \in \rho(Q^{[i,j]})$. We will refer to the index $t$ as representing moments of time, and we will then say that path $P^{[i,j]}$ was \emph{cut off} at time $t_j$.

For successive moments of time $t \in N$, we denote by $J_t$ the set of surviving path indices at time $t$, i.e., $J_t = \{ j : \forall_{t'\leq t}\ i_{t'} \notin \rho(Q^{[i,j]})\}$. We then obtain subforest $F_t \subseteq F$ by restricting $F$ to its surviving part, $F_t = \bigcup_{j \in J_t} P^{[i,j]} \cup \bigcup_{j \in J_t} Q^{[i,j]}$.

To prove the claim, we will consider how the random variable $M^+_i$ increases over time, until $F_t = \emptyset$. We again couple our sampling process by first deciding in each time step whether to place $i_t$ in forest $F\setminus F_{t-1}$ or in forest $F_{t-1}$, and only afterwards fixing for $i_t$ a specific free site with uniform probability within the chosen subforest. Observe that if $i_{t}$ is placed $F\setminus F_{t-1}$, then in the given step of the considered process, the value of $M^+_i$ remains unchanged at time $t$. We will thus eliminate from the process all time steps $t$ such that  $i_{t} \in F\setminus F_{t-1}$, and by a slight abuse of notation, we will relabel time indices as if these steps never occurred. Thus, in each time step $t$, we assume that a free site is picked for $i_t$ from $F_{t-1}$ uniformly at random.

Consider now the random variable $\phi_t = |J_{t-1}| - |J_t|$, representing the number of paths cut off in time step $t$. The expectation of $\phi_t$ can be lower-bounded, regardless of the history of the process.

\smallskip
\emph{Claim (*).} $\E [\phi_t | F_{t-1}] > 1/3$, for any $F_{t-1}\neq \emptyset$.

\smallskip
\emph{Proof (of Claim).} Fix forest $F_{t-1}$. We assign to each edge $e \in \bigcup_{j \in J_{t-1}} Q^{[i,j]}$ a weight $w(e) \in \N^+$, given as the number of $j \in J_{t-1}$ such that $P^{[i,j]}$ can be reached from $e$ by a descending path. For all edges $e \in \bigcup_{j \in J_{t-1}} P^{[i,j]}$, we put $w(e) = 0$. Now, if $e$ is chosen as the $t$-th edge in the process, we have $\phi_t = |J_{t-1}| - |J_t| = w(e)$. It follows that:
$$
\E [\phi_t | F_{t-1}] = \frac{\sum_{e\in F_{t-1}} w(e)}{|F_{t-1} \setminus \{e\in F : \rho(e) < i_t \}|} \geq \frac{\sum_{e\in F_{t-1}} w(e)}{|F_{t-1}|} \geq \frac{|J_{t-1}| (r^{[i]} - r^{[i-1]})}{|J_{t-1}| (r^{[i+1]} - r^{[i-1]})} > \frac{1}{3},
$$
which completes the proof of the claim.

\smallskip
Moreover, taking into that $\phi_t$ has bounded range $\phi_t \in [0,l_i]$, and that $\sum_{t\in \N} \phi_t = l_i$, we obtain a concentration result on the number of steps until the stopping of the process ($F_t = \emptyset$); for completeness, we provide a standalone proof.

\smallskip
\emph{Claim (**).} For $T = 12 c l_i$, we have: $\Pr [F_T \neq \emptyset] < 2^{-c}$.

\smallskip
\emph{Proof (of Claim).} We define the non-negative submartingale $\Phi'_T = \sum_{t=1}^T \phi'_t$ as follows. When $F_{t-1}\neq \emptyset$, we choose $\phi'_t \in [0,l_i]$ to be dominated by $\phi_t$, so that $\phi'_t \leq \phi_t$ and $\E [\phi'_t | F_{t-1}] = 1/3$. (The latter condition can always be satisfied by Claim (*)). When $F_{t-1} = \emptyset$, we put $\phi'_t = 1/3$. Observe that when $\Phi'_T > l_i$ for some $T$, we necessarily have $\phi'_t > \phi_t$ for some $t\leq T$, hence $F_{T} =\emptyset$. To lower-bound the probability of the event $\Phi'_T > l_i$, we remark that the bounds $\phi'_t \in [0,l_i]$ and $\E [\phi'_t | F_{t-1}] = 1/3$ imply the following bound on variance of the process: $\sigma^2 [\phi'_t | F_{t-1}] \leq l_i/3$. Now, using a standard martingale bound (cf. e.g.~\cite{concen.pdf}[Thm.~18] applied to the process $X_t = \Phi'_t - t/3$), we obtain for any $\lambda > 0$:
$$
\Pr [ \Phi'_T \leq T/3 - \lambda] \leq \exp\left(\frac{-\lambda^2/2}{T l_i/3 + \lambda l_i/3}\right).
$$
Substituting $T = 12 c l_i$ and $\lambda = 3 c l_i$, we obtain:
$$
\Pr [ \Phi'_T \leq c l_i] \leq e^{-0.9c} < 2^{-c}.
$$
Taking into account that $c\geq 1$ by assumption, the claim follows directly.

\smallskip
Recalling that in each time step $t$ with $F_t \neq \emptyset$, the value of random variable $M^+_i$ increases by at most $1$, we obtain directly from Claim (**) that $\Pr [M^+_i > 12 c l_i] < 2^{-c}$, which completes the proof.
\end{proof}

Next, for $2 |i$, let $C_i \in \N^+$ be a random variable defined as the smallest integer such that $M^+_i \leq 12 C_i l_i$. Since $C_i$ depends only on random values $\rho$ chosen with $L^{*[i-1]} \cup L^{*[i]}$, the random variables $\{C_i\}_{2|i}$ are independent. Moreover, by Lemma~\ref{lem:layer_geometric}, each $C_i$ may be stochastically dominated by a (independent) geometrically distributed random variable $\Gamma_i$ with parameter $p=1/2$. It follows that:
$$
\sum_{\substack{2|i,\\i\leq i_{max}}} C_i \leq \sum_{\substack{2|i,\\i\leq i_{max}}} \Gamma_i \sim \mathrm{NB}(\lfloor i_{max}/2 \rfloor,1/2),
$$
where the parameters of the negative binomial distribution $\mathrm{NB}(r,p)$ represent the number of trials with success probability $p$ until $r$ successes are reached. An application of a rough tail bound for $\mathrm{NB}(\lfloor i_{max}/2 \rfloor,1/2)$ gives:
\begin{equation}
\label{eq:Cbound}
\Pr\left[\sum_{2|i} C_i < 2i_{max} + 4c \ln n\right] > 1- n^{-c}, \text{\quad for any $c>1$}.
\end{equation}
Recalling that $M^+_i = \sum_{j=1}^{l_i}|M^{+[i,j]}| \leq 12 C_i l_i$, we may write by concavity of the logarithm function:
$$
\sum_{j=1}^{l_i} (1 + \ln^\circ |M^{+[i,j]}|) \leq l_i + l_i \ln^\circ \left(\frac{1}{l_i}\sum_{j=1}^{l_i}|M^{+[i,j]}|\right) \leq l_i + l_i \ln (12 C_i) < l_i \ln (33 C_i).
$$
We can therefore bound variable $N_{even}^+$, taking into account its definition \eqref{eq:N}:
\begin{equation}
\label{eq:Nbound}
N_{even}^+  = \sum_{\substack{2|i\\ 1\leq j \leq l_i}} (1+\ln |M^{+[i,j]}|) \leq \sum_{2|i} l_i \ln (33 C_i).
\end{equation}
We now apply a union bound over the two events given by~\eqref{eq:Abound} and~\eqref{eq:Cbound}, which hold w.h.p., From Eq.~\eqref{eq:Abound}, \eqref{eq:Cbound}, and~\eqref{eq:Nbound}, we have that for any $c>1$, the following event holds w.p.\ at least $1 - 2n^{-c}$:
\begin{equation}\label{eq:Nevenfinal}
N_{even}^+ \leq 3c \ln n + \max_{(c_i)}\sum_{\substack{2|i\\i\leq i_{max}}} l_i \ln (33 c_i),
\end{equation}
where $(c_i)_{i\leq i_{max}, 2|i}$ are positive integers satisfying the condition: $\sum_{2|i} c_i < 2i_{max} + 4c \ln n$.

Returning to Eq.~\eqref{eq:Xplus}, with respect to $N_{odd}^+$ an analogous technique gives us that w.p.\ at least $1 - 2n^{-c}$:
\begin{equation}\label{eq:Noddfinal}
N_{odd}^+ \leq 3c \ln n + \max_{(c_i)}\sum_{\substack{2\neq i\\i\leq i_{max}}} l_i \ln (33 c_i),
\end{equation}
where $(c_i)_{i\leq i_{max}, 2\not|\; i}$ are likewise positive integers satisfying the condition: $\sum_{2\not|\; i} c_i < 2i_{max} + 4c \ln n$.

Combining Eq.~\eqref{eq:Xplus} with Eq.~\eqref{eq:Nevenfinal} and~\eqref{eq:Noddfinal} through a union bound (and substituting $\gamma_i := 33c_i$), we eventually obtain that w.p.\ at least $1 - 4n^{-c}$:
\begin{equation}
\label{eq:Xplusfinal}
\sum_{\eta \in E(T^{**}_u)} X^+(\eta) \leq \sum_{r=0}^{19} |\cut_u^{*(r)}| +  6c \ln n + \max_{(\gamma_i)}\sum_{i\leq i_{max}} l_i \ln \gamma_i,
\end{equation}
where $(\gamma_i)_{i\leq i_{max}}$ satisfy the condition:
\begin{equation}
\label{eq:cond_gamma}
\forall_{i} \gamma_i \in \N^+ \quad \text{and} \quad \sum \gamma_i < 132 i_{max} + 264 c \ln n.
\end{equation}

\paragraph{Bounding the Sum of $X^-(\eta)$.}

For the random variables $X^-(\eta)$, the main arguments required to establish the bound are similar to those in the case of $X^+(\eta)$; we confine ourselves to an exposition of the differences. The main difference is that for a path $P^{[i,j]}$, instead of a unique predecessor path $Q^{[i,j]}$ in layer $L^{*[i-1]}$, we now have to deal with multiple possible descendant paths in layer $L^{*[i+1]}$; on the other hand, the outward-branching structure of the tree means that we can show tighter concentration bounds in this case.

We recall that $\mathcal{P}^{[i]}$ is the set of descending paths $P$ in $T_u$ which stretch precisely between the endpoints of the $i$-th layer, and for $1 \leq j \leq l_i$, we denote by $\mathcal R^{[i,j]}$ the set of all paths in $\mathcal{P}^{[i+1]}$ which are extensions of $P^{[i,j]}$, i.e., $\mathcal R^{[i,j]} \subseteq \mathcal{P}^{[i+1]}$ and for all $R \in \mathcal R^{[i,j]}$, we have $ P^{[i,j]} \subseteq P_u(R^+)$.  For the event $X^-(\eta)=1$ to hold, it is now necessary that two conditions are jointly fulfilled: $\eta$ must satisfy the \emph{suffix minimum} condition on the path $P^{[i,j]}$:
\begin{equation}\label{eq:suffix_min}
\eta = \arg\min_{e \in P^{[i,j]} \setminus P_u(\eta^+)} \rho(e),
\end{equation}
and moreover, $\rho(\eta) < \max_{R \in \mathcal R^{[i,j]}} \min \rho(R)$. We next denote by $M^{-[i,j]} \subseteq P^{[i,j]}$ the set of all edges $\eta \in P^{[i,j]}$ satisfying $\rho(\eta) <  \max_{R \in \mathcal R^{[i,j]}} \min \rho(R)$. We further denote by $\eta^{-[i,j,k]}$ the $k$-th edge in $M^{-[i,j]}$, when ordering edges of $M^{-[i,j]}$ by decreasing distance to the root $u$, $1\leq k \leq |M^{-[i,j]}|$. Finally, we denote by $Z^{-[i,j,k]} \in \{0,1\}$ the indicator random variable for the event that ``edge $\eta^{-[i,j,k]}$ satisfies the suffix minimum condition~\eqref{eq:suffix_min} on path $P^{[i,j]}$''.

The subsequent analysis proceeds as before, and we obtain direct analogues of Eq.~\eqref{eq:Xplus}, \eqref{eq:N}, and~\eqref{eq:Abound}, replacing superscripts ``$+$'' of all random variables by ``$-$''.

We next denote $M^-_i := \sum_{j=1}^{l_i}|M^{-[i,j]}|$, for $i\geq 2$, and obtain the following analogue of Lemma~\ref{lem:layer_geometric}.
\begin{lemma}\label{lem:layer_geometric_suffix}
Fix $i \geq 2$. Then, for any $c\geq 1$:
$$
\Pr \left[M^-_i > 12 c l_i \right] < e^{-4 c l_i} < 2^{-c}.
$$
\end{lemma}
\begin{proof}
The proof follows along similar lines as that of Lemma~\ref{lem:layer_geometric}.

For fixed $i$, let $\mathcal{R}^{[i]} := \bigcup_{j=1}^{l_i}\mathcal{R}^{[i,j]}$. We consider the edge set $F := \bigcup_{j=1}^{l_i} P^{[i,j]} \cup \mathcal{R}^{[i]}$ (forming a subforest of $T^*_u$, contained entirely within layers $L^{*[i]}$ and $L^{*[i+1]}$). Choose arbitrarily the set $I$ of (necessarily distinct) values of $\rho$ which appear within $F$, $I = \rho(F)$. Let  $I = \{i_1,i_2,\ldots, i_{|F|}$\}, with $i_1 \leq i_2 \leq \ldots \leq i_{|F|}$.

As in the proof of Lemma~\ref{lem:layer_geometric}, we couple the sampling of values $\rho$ on $F$ with the following two-stage process. First, we fix set $I$. Then, given a choice of $I$, we select a uniformly random permutation to perform the assignment of values from $I$ to edges in $F$ ($|I|=|F|$). The latter permutation is defined iteratively, by assigning to successive values $i_t$, $1\leq t \leq |F|$, an as yet unoccupied edge (site) from $F$. The value of $|M^{-[i,j]}|$ is given as the number of elements of $I$ which are placed in sites from $P^{[i,j]}$ before the smallest index $t\in \N$, such that for all paths $R\in \mathcal{R}^{[i,j]}$, there exists $t' < t$ such that $i_{t'} \in \rho(R)$. We will refer to the index $t$ as representing moments of time.

Let $\mathcal{R}^{[i]} = \{R_1, \ldots, R_{|\mathcal{R}^{[i]}|}$. We say that a path $R_k \in \mathcal{R}^{[i]}$ was \emph{cut off} at time $t$ if $t$ is the smallest time such that $i_t \in \rho (R_k)$. For successive moments of time $t \in N$, we denote by $K_t$ the set of surviving indices $k$ of paths $R_k$ which have not been cut off at time $t$. We obtain subforest $F_t \subseteq F$ by restricting $F$ to its surviving part, where we treat a path $P^{[i,j]}$ if it extends to at least one surviving path $R_k$:
$$
F_t = \bigcup_{k \in K_t} R_k \cup \bigcup_{\substack{j :\\ \exists_{k \in K_t}\ P^{[i,j]} \subseteq P_u(R_k^+)}} P^{[i,j]}.
$$
Exactly as in the proof of Lemma~\ref{lem:layer_geometric},  we will consider how the random variable $M^-_i$ increases over time, until $F_t = \emptyset$. We again couple our sampling process by first deciding in each time step whether to place $i_t$ in forest $F\setminus F_{t-1}$ or in forest $F_{t-1}$, and only afterwards fixing for $i_t$ a specific free site with uniform probability within the chosen subforest. Observe that if $i_{t}$ is placed $F\setminus F_{t-1}$, then in the given step of the considered process, the value of $M^-_i$ remains unchanged at time $t$. We will thus eliminate from the process all time steps $t$ such that  $i_{t} \in F\setminus F_{t-1}$, and by a slight abuse of notation, we will relabel time indices as if these steps never occurred. Thus, in each time step $t$, we assume that a free site is picked for $i_t$ from $F_{t-1}$ uniformly at random.

Consider now the random variable $\phi_t = |K_{t-1}| - |K_t|$, representing the number of paths $R$ cut off in time step $t$. The expectation of $\phi_t$ can be lower-bounded, regardless of the history of the process.

\smallskip
\emph{Claim (*).} $\E [\phi_t | F_{t-1}] \geq 1/2$, for any $F_{t-1}\neq \emptyset$.

\smallskip
\emph{Proof (of Claim).} Fix forest $F_{t-1}$. When inserting $i_t$, the number of free sites in layer $L^{*[i+1]}$ is at least:
$$
\sum_{k \in K_{t-1}} |R_k| \geq |K_{t-1}| (r^{[i+1]} - r^{[i]}).
$$
On the other hand, since each surviving path $P^{[i,j]}$ extends to some surviving path $R_k$, the total number of free sites for insertion of $i_t$ is upper-bounded by $ |K_{t-1}| (r^{[i+1]} - r^{[i-1]}$. Since insertion of $i_t$ into layer $L^{*[i]}$ means that $|K_t| =  |K_{t-1}|$, and insertion of $i_t$ into layer $L^{*[i+1]}$ means that  $|K_t| =  |K_{t-1}| -1$, we obtain:
$$
\E [\phi_t | F_{t-1}] \geq \frac{|K_{t-1}| (r^{[i+1]} - r^{[i]})}{|K_{t-1}| (r^{[i+1]} - r^{[i-1]})} \geq \frac{1}{2},
$$
which completes the proof of the claim.

\smallskip
Moreover, taking into that $\phi_t$ has bounded range $\phi_t \in \{0,1\}$, and that $\sum_{t\in \N} \phi_t \leq l_i$, we obtain a concentration result on the number of steps until the stopping of the process ($F_t = \emptyset$) directly from the Azuma-McDiarmid martingale inequality (cf.~e.g.~\cite{concen.pdf}[Thm.~16],~\cite{McDiarmid}). For parameter $T = 12 c l_i$, we obtain after some transformations:
$$
\Pr [F_{12 c l_i} \neq \emptyset] < e^{-4 c l_i}.
$$
Recalling that in each time step $t$ with $F_t \neq \emptyset$, the value of random variable $M^-_i$ increases by at most $1$, we obtain directly that $\Pr [M^-_i > 12 c l_i] < e^{-4 c l_i}$, which completes the proof.
\end{proof}

The rest of the argument proceeds as for the case of $X^+$, applying Lemma~\eqref{lem:layer_geometric_suffix} in place of Lemma~\eqref{lem:layer_geometric}. We eventually obtain the following analogue of Eq.~\eqref{eq:Xplus}: for any $c>1$, w.p.\ at least $1 - 4n^{-c}$:
\begin{equation}
\label{eq:Xminusfinal}
\sum_{\eta \in E(T^{**}_u)} X^-(\eta) \leq \sum_{r=0}^{19} |\cut_u^{*(r)}| +  6c \ln n + \max_{(\gamma_i)}\sum_{i\leq i_{max}} l_i \ln \gamma_i,
\end{equation}
where $(\gamma_i)_{i\leq i_{max}}$ satisfy condition~\eqref{eq:cond_gamma}.

\paragraph{Combining Bounds.} Introducing the bounds of Eq.~\eqref{eq:Xplusfinal} and~\eqref{eq:Xminusfinal} into Eq.~\eqref{eq:StoX} through a union bound we obtain the following statement: For any $c>1$, w.p.\ at least $1 - 8n^{-c}$:
\begin{equation}
\label{eq:Xdescriptive}
|S(u)| \leq 2\sum_{r=0}^{19} |\cut_u^{*(r)}| +  12c \ln n + 2 \max_{(\gamma_i)}\sum_{i\leq i_{max}} l_i \ln \gamma_i,
\end{equation}
where $(\gamma_i)_{i\leq i_{max}}$ satisfy condition~\eqref{eq:cond_gamma}.

Now, recalling the bounds on $l_i$ from Lemma~\ref{lem:path_decomposition}, the bound $i_{max} < 16 \ln \diam$, setting $c=2$, and applying a union bound over all vertices $u$, we obtain the main technical result of the Section. We present it first in its strongest form, and then provide two more useful corollaries.

\begin{theorem}\label{pro:technical}
With probability at least $1 - 8/n$, all nodes $u\in V$ satisfy the following bound on hub set size:
\begin{equation}\label{eq:technical}
S(u) \leq 2\sum_{r=0}^{19} |\cut_u^{*(r)}| +  24 \ln n + 2 \max_{(\gamma_i)}\sum_{i=1,2,\ldots,\lfloor 16 \ln \diam\rfloor}  l_i \ln \gamma_i,
\end{equation}
where $l_i \leq 2\min_{r \in [r^{[i+1]},r^{[i+2]}]} |\cut_u^{*(r)}|$ with $r^{[i]}= \left\lceil 15 \left( \frac{16}{15}\right)^i\right\rceil$, and the maximum is taken over positive integers $(\gamma_i)$ satisfying the condition: $\sum \gamma_i < 2112 \ln \diam + 528 \ln n$.
\end{theorem}

We provide two more convenient corollaries of Theorem~\ref{pro:technical}. For the case when the considered trees $T_u$ are close to scale-free, we simply bound the size of all cuts $\cut_u^{*(r)}$ through skeleton dimension: $|\cut_u^{*(r)}| \leq \sk$. Bound~\eqref{eq:technical} then takes the asymptotic form, for $i_{max} = \lfloor 16 \ln \diam \rfloor$:
\begin{equation}
\label{eq:su_bound}
S(u) \leq O(\sk) +  O(\log n) + O(\sk) \max_{(\gamma_i)}\sum_{i \leq i_{max}}  \ln \gamma_i,
\end{equation}
where the latter sum can be bounded using the concavity of the logarithm function as:
\begin{align*}
\max_{(\gamma_i)}\sum_{i \leq i_{max}} \ln \gamma_i &\leq i_{max} \max_{(\gamma_i)} \ln \left(\frac{1}{i_{max}}\sum_{i \leq i_{max}} \gamma_i\right) \leq i_{max} O\left(\max\left\{1,\log \frac{\log n}{\log \diam}\right\}\right) \leq\\
&\leq O\left(\log \diam \max\left\{1,\log \frac{\log n}{\log \diam}\right\}\right).
\end{align*}
We also observe the following link between the parameters $\sk$, $\diam$, and $n$. Since by Proposition~\ref{pro:doubling} graph $G$ has doubling dimension bounded by $2\sk + 1$, it follows that a radius-$\diam$ ball in $G$ may only contain at most $(2\sk+1)^{\lceil\log_2 \diam\rceil}$ nodes from $V$. Hence, $(2\sk+1)^{\lceil\log_2 D\rceil} \geq n$, and we obtain:
\begin{equation}\label{eq:nDk}
\log n = O(\log \sk \cdot \log \diam).
\end{equation}
Thus, the $O(\log n)$ additive factor in the bound~\eqref{eq:su_bound} on $S(u)$ is dominated in notation by the last factor of the sum, which is stated as at least $O(\sk \log \diam$).

Combining the above, we obtain the following corollary.
\begin{corollary}\label{cor:bound1}
With probability at least $1 - O(1/n)$, the hub set size of every node is bounded by: $$O\left(\sk \log \diam \max\left\{1,\log \frac{\log n}{\log \diam}\right\}\right).$$
\end{corollary}
In particular, when the graph has sufficiently large diameter, $\diam = n^{\Omega(1)}$, we have that the hub set size of all nodes is bounded by $O(\sk \log \diam)$. For the general case, by introducing Eq.~\eqref{eq:nDk} into Corollary~\ref{cor:bound1}, we obtain the following statement.
\begin{corollary}\label{cor:klabelingworstcase}
With probability at least $1 - O(1/n)$, the hub set size of every node is bounded by: $$O\left(\sk \log \log \sk \log \diam\right).$$
\end{corollary}

When considering the case of trees in which the width of tree $T^*_u$ is far from uniform over different scales of distance, tighter bounds are obtained by relating the size of $S(u)$ to the integrated skeleton dimension $\isk(u)$. To do this, we apply in Eq.~\eqref{eq:technical} the rough bound: $\ln \gamma_i < \ln \sum_i \gamma_i = O(\log \log n + \log \log \diam)$. This leaves us with an expression of the form:
$$
S(u) \leq O(\isk(u)) +  O(\log n) + O(\log \log n + \log \log \diam) \sum_{i\leq i_{max}}  l_i \leq O(\log n + \isk(u) (\log \log n + \log \log \diam)),
$$
where we used the bound $\sum_{i\leq i_{max}}  l_i \leq O(\isk(u))$, which follows easily from the definition of the parameter $\isk$.

\begin{corollary}\label{cor:isk}
With probability at least $1 - O(1/n)$, the hub set size of every node $u\in V$ is bounded by $O(\log n + \isk(u) (\log \log n + \log \log \diam)).$
\end{corollary}

\section{An Application to \texorpdfstring{$\DD$}{D}-preserving Distance Labeling}\label{sec:dpreserving}

As a slight extension of our results, we note that our technique based on analyzing tree skeletons for shortest path trees has direct application the $\DD$-preserving distance labeling problem in unweighted graphs, for some parameter $\DD>0$. We recall that a scheme is called $\DD$-preserving if for any queried pair of nodes $(u,v)$ with $\dist(u,v) > \DD$, the value returned by the decoder is equal to $\dist(u,v)$.

By analogy to the integrated skeleton dimension given by~\eqref{eq:isk_def}, we introduce a variant for this parameter which only considers cuts at distance more than $\DD/6$.

\begin{equation}
\label{eq:isk_def_D}
\isk_\DD(u) :=\sum_{v \in V_u^* : d(u,v) > \DD/6} \frac{1}{d_u(v)} = \sum_{r\in \N, r > \DD/6} \frac{|\cut_u^{*(r)}|}{r}.
\end{equation}

The claims of Lemma~\ref{lem:expected_hub_precise} and Corollary~\ref{cor:isk}, which give bounds on average hub size of $O(\isk(u))$ and $O(\log n + \isk(u) (\log \log n))$ in the unweighted setting, naturally translate to $\DD$-preserving labeling. For our techniques to be directly applicable, it suffices to subdivide each edge of the graph into a path of 12 vertices so that all distances between $u-v$ pairs are divisible by $12$, and to choose shortest path trees so that for any pair of nodes $u, v$, the intersection $T_u \cap T_v$ contains a shortest $u-v$ path (this may be achieved, for example, by enforcing a unique choice of shortest paths between any node pair, e.g., by choosing i.a.r. the length of each edge in the range $[1-1/n, 1])$. Then, the entire analysis holds, and we can eventually replace $\isk(u)$ by $\isk_\DD(u)$ in the statement of the claims.

We remark that it is an elementary property of the tree skeleton that $|\cut_u^{*(r)}| = O(n/r)$, since any node at distance $r$ from $u$ continues in $T_u$ along an independent path of length at least $r/2$. By performing the latter sum in~\eqref{eq:isk_def_D}, we obtain $\isk_\DD(u) = O(n/\DD)$. Thus, we obtain the following Proposition.

\begin{proposition}\label{pro:smalldpreserving}
There exists a hub labeling scheme for the $\DD$-preserving distance labeling problem with hubs of average size $O(n/\DD)$ and worst case size $O(\log n + (n/\DD) \log \log n)$. The size of the bit representations of the corresponding labels is  $O((n/\DD) \log n)$ and $O(\log^2 n + (n/\DD) \log n \log \log n)$, respectively.
\end{proposition}
The size of the obtained $\DD$-preserving hub-based labeling scheme is (almost) optimal, since there holds a lower bound of $\Omega(n/\DD)$ on both the average and worst-case size of hub sets~\cite{Sublinear}. In fact, our scheme can be modified slightly to obtain hub sets of worst-case size $O(n/\DD)$ up to a certain threshold value $\DD = \widetilde O(\sqrt n)$. We present the details of this modified scheme in the following Subsection.

\subsection{A Modified \texorpdfstring{$\DD$}{D}-preserving Labeling Scheme}

In this section we present an independent family of distance labeling schemes, which have the $\DD$-preserving property.   Whereas the scheme and the presented analysis are valid for any value of parameter $\DD>0$, we obtain an improvement on the previously discussed scheme only up to some threshold value $\DD = \widetilde O(\sqrt n)$.

\paragraph{Construction of the Labeling.}

Fix the value of parameter $\DD>0$, with $12|\DD$. The basic building block of our labeling is a construction of hub sets $S_\DD(u)$ for each node $u\in V$, which allow us to handle distance queries for pairs of nodes whose distance is in the range $[\DD, 5\DD/4]$.

Before providing the details of the constructions of sets $S_\DD(u)$, we first introduce some auxiliary notation. As before, for a pair of nodes $u,v\in V$, we denote by $P(u,v)$ a fixed shortest path between $u$ and $v$. In the definition of $P(u,v)$, ties between different shortest paths should be broken in a consistent manner over the whole graph, so that $P(u,v) = P(v,u)$, and the set of shortest paths rooted at a node $u$, $\bigcup_{v\in V} P(u,v)$, is a spanning tree of $G$.

\begin{figure}[t]
  \centering
    \vspace*{-2mm}
    \includegraphics[scale=1,trim={4.7cm 17.3cm 8.6cm 6.4cm},clip]{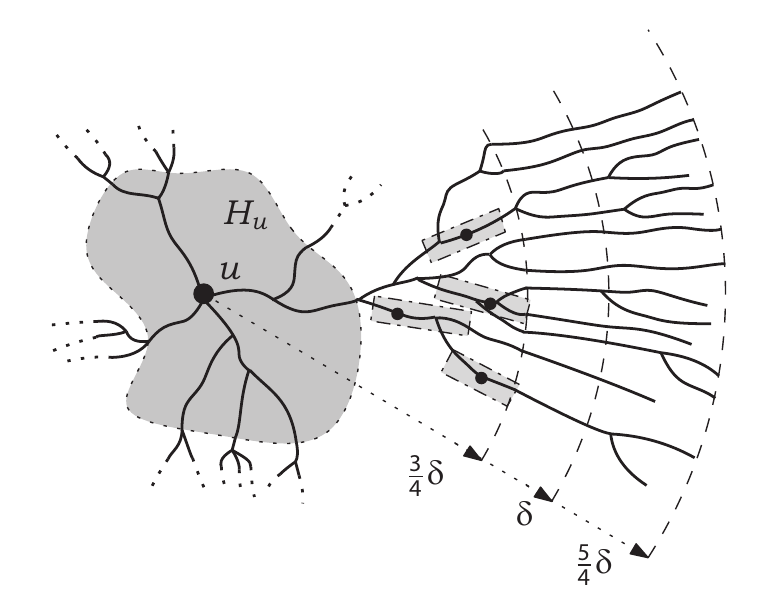}
    \vspace*{-2mm}
    \caption{Hub set selection for distance range $[\DD, 5\DD/4]$. The corresponding shortest path tree $T_u$ for some vertex $u$ is shown in the figure. The set of heavy vertices $H_u$ is shaded around vertex $u$; all remaining vertices up to distance $\frac34\DD$ belong to $L_u$. Vertices of $L'_u \subseteq L_u$ are also marked, with corresponding descending paths $P_d$ shaded.}
    \label{fig:randomized_hubs}
\end{figure}

For a node $u\in V$, we denote by $T_u$ the shortest path subtree of $G$, rooted at $u$, leading from $u$ to nodes at distance in the range $[\DD, 5\DD/4]$:
\[T_u = \bigcup_{v \in V \colon \cost(u,v)\in [\DD,5\DD/4]} P(u,v).\]
We denote by $T_u^*$ the subtree (skeleton) of tree $T_u$, also rooted at $u$ and truncated to its first $3\DD/4$ levels from the root: $T_u^* = T_u [\{v\in V(T_u) : \cost(u,v) \leq 3\DD/4\}]$. We remark that all descending paths in $T_u^*$ have reach of at least $\DD/4$ in $T_u$.

The set $S_\DD(u)$ will be constructed similarly as before, to include vertices from the central part of any $u-v$ path in the tree, for vertices $v$ with $\cost(u,v)\in [\DD,5\DD/4]$. However, we wish to control the number of possible bad events in which a descending path in the tree $T_u^*$ branches out at some level into too many subpaths, from each of which some representative node will need be chosen into $S_\DD(u)$. To do this, we will partition the vertex set of tree $T_u^*$ into two subsets, $H_u\cup L_u$, known as heavy and light vertices, respectively. A vertex $w$ of $T_u^*$ belongs to $H_u$ if the subtree of $T_u^*$ rooted at $w$ has at least $\DD$ leaves (all in the last level $3\DD/4$), and belongs to $L_u$ otherwise. We remark that $T_u^*[H_u]$ is a (possibly empty) subtree of $T_u^*$, whereas $T_u^*[L_u]$ is a sub-forest of $T_u^*$, in which each connected component is a tree with less than $\DD$ leaves. In all the considered trees, we maintain the same ancestry relation. In particular, we speak of a \emph{descending subpath} in a tree if one of its endpoints is an ancestor of the other with respect to the tree $T_u^*$ rooted at $u$.

We are now ready to define the hub sets $S_\DD(u)$, $u\in V$ by the following randomized construction. Assign to each node $v\in V$ a real value $\rho(v) \in [0,1]$, uniformly and independently at random. We now put for all $u\in V$:
\begin{equation}\label{eq:sdu}
S_\DD(u) := H_u \cup L'_u,
\end{equation}
where $L'_u \subseteq L_u$ is defined as the set of all vertices $v\in L_u$, such that there exists a descending subpath $P_d$ in $T_u^*[L_u]$, such that $v\in P_d$, $|P_d| = \DD/12$, and $v$ has the minimal value of $\rho$ along path $P_d$, $v = \arg\min_{w\in P_d} \rho(w) \equiv \arg\min \rho(P_d)$. See Fig.~\ref{fig:randomized_hubs} for an illustration.


\paragraph{Correctness.}

We start by showing that sets $S_\DD$ have the desired hub property, regardless of the choice of random values $\rho$ (which may only affect the size of these sets).

\begin{lemma}\label{lem:dpreserving_correctness}
For any pair of nodes $u,v \in V$ such that $\cost(u,v) \in [\DD, 5\DD/4]$, we have:
\[\cost(u,v) = \min_{w\in S_\DD(u)\cap S_\DD(v)} \left( \cost(u,w) + \cost(w,v)\right).\]
\end{lemma}

\begin{proof}
Consider the path $P = P(u,v) = P(v,u)$. We have $P \subseteq T_u$ and $P \subseteq T_v$. Moreover, $|P| \leq 5\DD/4$ and $|P \cap T_u^*| = |P \cap T_v^*| = 3\DD/4$. Denoting by $P^*$ the subpath of $P$ which belongs to both trees $T_u^*$ and $T_v^*$, $P^* = P \cap T_u^* \cap T_v^*$, it follows that $|P^*| \geq \DD/4$. 
We now prove the claim of the lemma by showing that at least one vertex from $P^*$ has to belong to $S_\DD(u) \cap S_\DD(v)$. We achieve this by a case analysis, depending on the portions of path $P^*$ which belong to the sets $L_u, H_u, L_v$, and $H_v$.
\begin{itemize}
\item If $|P^* \cap H_u \cap H_v| > 0$, then there exists at least one vertex $w \in P^* \cap H_u \cap H_v \subseteq P^* \cap S_\DD(u)\cap S_\DD(v)$, which completes the proof.
\item If $|P^* \cap L_u \cap H_v| \geq \DD/12$, then there exists at least one descending subpath $P_d$ of length $\DD/12$ in $T_u^*[L_u]$ which is completely contained in $P^* \cap H_v$. Setting $w = \arg\min \rho(P_d)$, we have $w \in L'_u$, and so it follows that $w \in P^* \cap L'_u \cap H_v \subseteq P^* \cap S_\DD(u)\cap S_\DD(v)$.
\item If $|P^* \cap H_u \cap L_v| \geq \DD/12$, we obtain the result by applying analogous considerations as in the previous case.
\item Finally, in all other cases we must have $|P^* \cap L_u \cap L_v| \geq \DD/12$. It follows that there exists at least one subpath $P_d \subseteq P^*$ of length $\DD/12$, which is a descending subpath in both $T_u^*[L_u]$ and $T_v^*[L_v]$. Setting $w = \arg\min \rho(P_d)$, we obtain $w \in L'_u$ and $w \in L'_v$, hence $w \in P^* \cap L'_u \cap L'_v \subseteq P^* \cap S_\DD(u)\cap S_\DD(v)$.
\end{itemize}
\end{proof}


\paragraph{Analysis.}
We now consider the size of sets $S_\DD(u)$. The size of set $H_u$ is independent of the choice of random variables $\rho$; it can easily be bounded, taking into account that tree $T_u^*$ has $\bigo{n/\DD}$ leaves.
\begin{lemma}\label{lem:heavy}
For all $u\in V$, $|H_u| \leq 3 n /\DD$.
\end{lemma}

\begin{proof}
Let $l \in H_u$ be a leaf node of $T_u^* [H_u]$. By definition of $H_u$, we have that the subtree of $T_u^*$ rooted at $l$ has at least $\DD$ leaves. As every leaf of $T_u^*$ is at depth $3\DD/4$, and all leaves in tree $T_u$ are at depth at least $\DD$, it follows that the subtree of $l$ in $T_u$ contains at least $\DD$ disjoint descending paths of length $\DD/4$ each, and so it has at least $\DD^2/4$ nodes. Since the size of tree $T_u$ is at most $n$, we obtain that tree $T_u^* [H_u]$ has at most $4n / \DD^2$ leaves. Moreover, the distance of each node of $T_u^* [H_u]$ from its root $u$ is at most $3\DD/4$. Hence, $|T_u^*[H_u]| \leq (3\DD/4)(4n / \DD^2) = 3n/\DD$.
\end{proof}

The size of set $L'_u$ depends on the choice of random variables $\rho$. We start by bounding the expected number of elements of $L'_u$, belonging to specific connected components of $T_u^*[L_u]$. Suppose $T_u^*[L_u]$ be a forest consisting of $k_u$ trees, and let $L_u = L_u^{(1)} \cup \ldots \cup L_u^{(k_u)}$ be the partition of its vertex set such that $T_u^*[L_u^{(i)}]$ represents its $i$-th connected component. Let $l_u^{(i)}$ denote the number of leaves of tree $T_u^*[L_u^{(i)}]$. Finally, let $L_u'^{(i)} = L_u' \cap L_u^{(i)}$. Clearly, $L_u'^{(1)} \cup \ldots \cup L_u'^{(k_u)}$ is a partition of $L_u'$. In the following, we consider the random variable $|L_u'| = \sum_{i=1}^{k_u} |L_u'^{(i)}|$, showing that it has an expectation of $\bigo{n/\DD}$, and obtaining concentration results around this expectation.

First, we remark that each descending path of tree $T_u^*$ contributes $\bigo{1}$ elements in expectation to set $T_u^*$. Consequently, the expected size of set $|L_u'^{(i)}|$ can be related to the number of leaves in the considered connected component.
\begin{lemma}\label{lem:exp_component}
For all $u\in V$ and all $1\leq i \leq k_u$, $\E |L_u'^{(i)}| \leq 36 l_u^{(i)}$.
\end{lemma}

\begin{proof}
Let $\P_u^{(i)}$ be the set of (inclusion-wise) maximal descending paths in the tree $T_u^*[L_u^{(i)}]$. We remark that $|\P_u^{(i)}| \leq l_u^{(i)}$.

For a path $P \in \P_u^{(i)}$, let $M_P(v)$ be the event that there exists a subpath $P_d \subseteq P$, with $|P_d| = \DD/12$, such that $v = \arg\min \rho(P_d)$. We have $\Pr[M_P(v)] = 0$ for $v\notin P$. If $v \in P$, then we use the following folklore probability estimation: for $M_P(v)$ to hold, one of the two (at most) descending subpaths $P'$ of $P$ of length $\DD/24$, having $v$ as one of their endpoints, must satisfy $v = \arg\min \rho(P')$. It follows that for $v \in P$, we have $\Pr[M_P(v)] \leq 48 / \DD$. By linearity of expectation we now obtain a bound on $\E | L_u'^{(i)}| $:
\[
\E| L_u'^{(i)} | = \sum_{v\in V}\Pr[v \in L_u'^{(i)}] \leq \sum_{P\in \P_u^{(i)}} \sum_{v\in P} \Pr[M_P(v)] \leq |\P_u^{(i)}| \cdot \max_{P\in \P_u^{(i)}} |P| \cdot \frac{48}{\DD} \leq l_u^{(i)} \cdot \frac{3\DD}{4} \cdot \frac{48}{\DD} = 36 l_u^{(i)}.
\]
\end{proof}

By linearity of expectation, we can apply the claim of Lemma~\ref{lem:exp_component} over all connected components, obtaining the following result.
\begin{lemma}\label{lem:exp_all}
For all $u\in V$, $\E |L_u'| \leq 144 n/\DD$.
\end{lemma}
\begin{proof}
By Lemma~\ref{lem:exp_component} we have:
\[
\E |L_u'| = \E\sum_{i=1}^{k_u} |L_u'^{(i)}| \leq 36 \sum_{i=1}^{k_u} l_u^{(i)}.
\]
The claim follows when we observe that $\sum_{i=1}^{k_u} l_u^{(i)} \leq 4n /\DD$. Indeed, this sum represents the total number of leaves in $T_u^*[L_u]$. Each such leaf (located at distance $3\DD/4$ from $u$) is the upper endpoint of a distinct descending path of length at least $\DD/4$ in the tree $T_u$, and $|T_u| \leq n$, hence we obtain the bound.
\end{proof}

In order to apply Chernoff bounds to the sum of random variables $|L_u'|$, we start by bounding the range of these variables.

\begin{lemma}\label{lem:lu_simple_bound}
For all $u\in V$ and all $1\leq i \leq k_u$, $|L_u'^{(i)}| \leq |L_u^{(i)}| < \DD^2$.
\end{lemma}

\begin{proof}
We have $L_u'^{(i)} \subseteq L_u^{(i)}$. By the definition of set $L_u$, tree $T_u^*[L_u^{(i)}]$ has less than $\DD$ leaves and all its nodes are at distance at most $3\DD/4$ from its root. It follows that $|L_u^{(i)}| < \DD\cdot 3\DD/4 < \DD^2$.
\end{proof}

The above upper bound provides an estimate on the maximum value of each random variable $|L_u'^{(i)}|$. However, in order to be able to perform a concentration analysis in a range of fairly large $\DD$ (roughly, for $n^{1/3} < \DD < n^{0.5}$), we also need to bound more tightly the concentration of each  $|L_u'^{(i)}|$ around its expected value.

Let $V_0 = \{v \in V : \rho(v) < 50 \frac{\ln n}\DD\}$. We start by showing that with high probability, all elements of the sets $|L_u'|$ belong to $V_0$.
\begin{lemma}\label{lem:bad_x}
Denote by $X$ the ``bad'' event that there exists a node $u\in V$, such that $L_u' \not\subseteq V_0$. We have: $\Pr[X] < 1/n$.
\end{lemma}

\begin{proof}
Consider first the probability $p_v$ that a fixed node $v\in V\setminus V_0$ satisfies $v = \arg\min \rho(P)$, where $P$ is any fixed path of $\DD/12$ nodes in $G$ which contains $v$. We have (with the last inequality holding when $\DD \ge 300$):
\[
p_v = \prod_{w\in P \setminus\{v\}} (1 - \rho(v)) = (1-\rho(v))^{\DD/12 -1} < (1 - 50\ln n /\DD)^{\DD/12-1} < n^{-4}.
\]

The probability of event $X$ occurring can be upper-bounded by performing a union bound over all nodes $u$, all descending paths $P$ in $T_u^*[L_u]$, and all nodes $v\in P$ of the event $[ v = \arg\min \rho(P) \cap  v\notin V_0]$ occurring. For each node $u$, there are at most $n$ such paths to consider, and less than $\DD$ possible nodes $v$ in each path. Overall, we obtain:
\[
\Pr[X] \leq n \cdot n \cdot \DD \cdot p_v < n^3 \cdot n^{-4} = 1/n.
\]
\end{proof}

Next, we show that with high probability, each connected component $L_u^{(i)}$ contains at most $\bigo{\DD \log n}$ nodes from $V_0$.
\begin{lemma}\label{lem:bad_y}
Denote by $Y$ the ``bad'' event that there exists a node $u\in V$ and $1\leq i \leq k_u$, such that $|L_u^{(i)} \cap V_0| > 100 \DD \ln n$. We have: $\Pr[Y] < 1/n$.
\end{lemma}

\begin{proof}
Let $Z(v)$ denote the indicator variable for node $v$ and set $V_0$, i.e., $Z(v)=1$ if $v\in V_0$ and $Z(v)=0$ otherwise. Clearly, $\Pr[Z(v)=1] = 50 \ln n / \DD$, and all random variables $Z(v)$, $v\in V$ are independent.
For any fixed $|L_u^{(i)}|$, we have:
\[
\E \sum_{v\in L_u^{(i)}} Z(v) \leq |L_u^{(i)}| \cdot 50 \ln n / \DD \leq 50 \DD \ln n,
\]
where we used~\ref{lem:lu_simple_bound} to bound $|L_u^{(i)}|$. Next, we proceed by apply a simple multiplicative Chernoff bound for the considered random variable:
\[
\Pr[|L_u^{(i)} \cap V_0| > 100 \DD \ln n] = \Pr[\sum_{v\in L_u^{(i)}} Z(v) > 100 \DD \ln n] \leq e^{-50 \DD\ln n / 3} < n^{-16}.
\]
Applying a union bound over $L_u^{(i)}$, for all $u\in V$ and $1\leq i \leq k_u$, gives the claim.
\end{proof}

We are now ready to apply a Chernoff-bound type analysis to the random variable $|L_u'|$.
\begin{lemma}\label{lem:light_small_d}
Let $\DD \leq \sqrt n / \ln n$. Then: $\Pr[\forall_{u \in V}\ |L_u'| < 800 n/\DD] = 1 - \bigo{1/n}$.
\end{lemma}

\begin{proof}
Define the random variable $\lambda_u'^{(i)}$ as $|L_u'^{(i)}|$ when $L_u'^{(i)} \subseteq L_u^{(i)} \cap V_0$ and $|L_u^{(i)} \cap V_0| \leq 100 \DD\ln n$, and fix $\lambda_u'^{(i)} = 0$ otherwise.

Define $\lambda_u' = \sum_{i=1}^{k_u} \lambda_u'^{(i)}$. All $\lambda_u'^{(i)}$ are independent random variables, since they are functions of disjoint sets of random variables $(\rho(v) : v \in  L_u^{(i)})$. Moreover, we have $0\leq  \lambda_u'^{(i)} \leq 100 \DD \ln n$. An application of a simple multiplicative Chernoff bound gives for any $A>0$:
\[
\Pr[\lambda_u' \geq A ] < \left(e\frac{\E \lambda _u'}{A}\right)^{A/(100 \DD \ln n)} \leq \left(\frac{144e n/\DD}{A}\right)^{A/(100 \DD \ln n)} < \left(\frac{400 n/\DD}{A}\right)^{A/(100 \DD \ln n)},
\]
where we used the bound $\E \lambda_u' \leq \E |L_u'| \leq 144 n/\DD$ following from Lemma~\ref{lem:exp_all}. Putting $A = 800 n/\DD$ and taking into account that $\DD \leq \sqrt n / \ln n$, we get for sufficiently large $n$:
\[
\Pr[\lambda_u' \geq 800 n/\DD ] < 2^{-8 n/(\DD^2 \ln n)} < \frac{1}{n^2}.
\]
By applying a union bound over all nodes, we obtain that $\Pr[\forall_{u \in V}\ \lambda_u' < 800 n/\DD] > 1 - 1/n$. Taking into account Lemmas~\ref{lem:bad_x} and~\ref{lem:bad_y}, we also have:
\[
\Pr[\forall_{u \in V}\ \lambda_u' = |L_u'| ] > 1 - 2/n.
\]
Overall, we obtain:
\[
\Pr[\forall_{u \in V}\ |L_u'| < 800 n/\DD] > 1 - 3/n.
\]
\end{proof}

In view of the definition of the proposed hub set labeling (Eq.~\eqref{eq:sdu}), Lemmas~\ref{lem:heavy} and~\ref{lem:light_small_d} complete the analysis of the case of $\DD\leq \sqrt n / \ln n$, showing that our randomized construction yields with high probability hub sets of size $\bigo{n/\DD}$ for all nodes of the graph.

\begin{proposition}\label{pro:largedpreserving}
For any $0 < \DD \leq \sqrt n / \ln n$, $12|\DD$, there exists a hub labeling scheme which correctly decodes the distance between any pair of nodes lying at a distance in the range $[\DD, 5\DD/4]$, using hubs of size at most $O(n/\DD)$.\qed
\end{proposition}

\subsection{Improved \texorpdfstring{$\DD$}{D}-preserving Distance Labeling for Arbitrary Distance}

For an arbitrary instance of the $\DD$-preserving distance labeling problem, we can now construct a hub set $S^+(u)$ by combining the results of Propositions~\ref{pro:largedpreserving} and~\ref{pro:smalldpreserving}, for large and small scales of distance, respectively. Formally, we put:
\begin{equation}
\label{eq:hub_sum}
S^+(u) := S_{\DD_{\max}}(u) \cup \bigcup_{\substack{i=0,1,2,\ldots\\ \DD_i < \DD_{\max}}} S_{\DD_i} (u),
\end{equation}
where in the first part of the expression, the value
$$
\DD_{\max} = \max\{\DD, \lfloor \sqrt n / \ln n \rfloor\}
$$
is a suitably chosen threshold parameter, and the hub set $S_{\DD_{\max}}(u)$ is constructed following Proposition~\ref{pro:largedpreserving}, thus providing a $\DD_{\max}$-preserving distance labeling. In the second part of the expression, we take care of smaller distances from the range $[\DD,\DD_{\max})$, by applying Proposition~\ref{pro:smalldpreserving} over a specifically chosen distance sequence $\{\DD_i\}_i$, to obtain hub sets $S_{\DD_i}(u)$, such that each set $S_{\DD_i}(u)$ intersects with a shortest $u-v$ path, for all nodes $v$ such that $\cost(u,v) \in [\DD_i, 5 \DD_i/4]$. To obtain a coverage of the entire distance range $[\DD,\DD_{\max})$, we put $\DD_0 = 12 \lfloor \DD / 12 \rfloor < \DD$, and choose $\DD_{i+1}$ as the largest integer such that $12 | \DD_{i+1}$ and $\DD_{i+1} < 5\DD_i/4$. Since the sequence $\{\DD_i\}_i$ is geometrically increasing, in view of Proposition~\ref{pro:smalldpreserving}, we obtain the following bound:
\begin{equation}
|S^+(u)| = O(n / \DD + |S_{\DD_{\max}}(u)|).
\end{equation}
Now, taking into account the definition of $\DD_{\max}$ and bounding $|S_{\DD_{\max}}(u)|$ by Proposition~
\ref{pro:largedpreserving}, we directly obtain the main result of this Section.

\begin{theorem}\label{thm:dpreserving}
There exists a $\DD$-preserving distance labeling scheme based on hub sets, such that:
\begin{itemize}
\item[(i)] When $\DD\leq \sqrt n / (\ln n \ln \ln n)$, the hub sets of all nodes are of size $\bigo{n/\DD}$, which corresponds to distance labels of size $\bigo{n \log \DD/\DD}$ per node.
\item[(ii)] For any $\DD>0$, the hub sets of all nodes are of size $\bigo{\log \DD + n \log \log \DD/\DD}$, which corresponds to distance labels of size $\bigo{\log^2 \DD + n \log \DD \log \log \DD/\DD}$ per node.
\item[(iii)] For any $\DD>0$, the average size of a hub set, taken over all nodes, is $\bigo{n/\DD}$, which corresponds to distance labels of average size $\bigo{n \log \DD /\DD}$ per node.
\end{itemize}
Furthermore, the corresponding labels can be constructed in expected polynomial time.\qed
\end{theorem}

\section{Computing Skeleton Dimension and Distance Labels} \label{sec:compute}

\paragraph{Discrete skeleton representation.}
Given a tree $T$ rooted at node $u$ with length function $\ell$,
a discrete representation of its skeleton $T^*$ can be obtained
as the sub-tree with edges $vw\in E(T)$ such that
$\reach_T(v) \ge \frac{1}{2} d_T(u,v)$ equipped with length function $\ell'$
defined by $\ell'(vw)=\ell(vw)$ if $\reach_T(w) \ge \frac{1}{2} d_T(u,w)$ and
$\ell'(vw) = r - d_T(u,v)$ with $r = \frac{2}{3} (d_T(u,w) + \reach_T(w))$
otherwise. The idea is that node $w$ is a leaf of $T^*$
that corresponds to the point $w'$ of edge $vw$ in $\T$
that satisfies $\reach_{\T}(w') = \frac{1}{2} d_{\T}(u,w')$.
To see this, let $x$ be a descendant of $w$ such $d_T(w,x)=\reach_T(w)$. We then 
have $r=\frac{2}{3}d_T(u,x)$. We thus get
$\reach_{\T}(w')=d_T(v,x)-\ell'(vw)=d_T(u,x)-r=\frac{1}{3} d_T(u,x)=\frac{r}{2}$
whereas $d_{\T}(u,w')=d_T(u,v)+\ell'(vw)=r$.

\paragraph{Skeleton dimension computation.}
Given a tree $T$, the reach of each node can be computed by a scan of vertices in reverse topological order. Obtaining the discrete skeleton representation is then straight-forward. Its width $\sk$ can be computed by scanning vertices by non-decreasing distance from the root using a priority queue for storing edges containing nodes in $\cut_r(T)$. This can be done in $O(n\log\log \sk)$ time using a dedicated integer priority queue~\cite{thorup04}.

The skeleton of a graph can thus be simply obtained from an all pair shortest path computation. With integer lengths and dedicated priority queues~\cite{thorup04}, this can be done in $O(nm+n^2\log\log n)$ time.

We remark that faster computation of tree skeletons could be obtained in practice by using classical heuristics for bounding reach of nodes~\cite{reach,RE}. The algorithm proposed in \cite{reach} alternates partial shortest-path tree computation with introduction of shortcuts to obtain efficiently reach bounds on the graph plus the added shortcuts. The computation of partial trees up to a given radius $2r$ allows to prune nodes with reach less than $r$. Shortcuts allow to reduce reach of nodes with degree 2: if a node $v$ has two neighbors $u,w$, a shortcut $uw$ with length $\ell(uv)+\ell(vw)$ is added to by-pass $v$. The algorithm results in reach bounds on the graph with shortcuts. Reach bounds on the original graph can then be obtained by removing shortcuts in reverse order and updating the reach bound $R(v)$ of a node $v$ by-passed by shortcut $uw$ as $R(v):=\max(R(v),\min(R(u)+\ell(uv),\ell(vw)+R(w)))$ where $R(u)$ and $R(w)$ denote the reach bounds obtained for $u$ and $w$ respectively. A subtree containing the tree skeleton of a node $u$ can then be obtained through a partial Dijkstra from $u$ where we prune nodes whose reach is known to be less than half of the current distance from $u$. In practice, we believe that this would allow to compute each skeleton tree in time comparable to an RE query. Our labeling algorithm can then be adapted to  take the resulting family of trees as input.

\paragraph{Distance label computation.}
Computing the hub set of a tree $T$ is more intricate as we have to emulate the subdivision of each edge of length $a$ into $12a$ unweighted edges to conform to the analysis of Section~\ref{sec:hub_labeling}.
For the sake of notation, we number these unweighted edges of the subdivision from 1 to $12a$, and let $\rho_i$, $1 \leq i \leq 12a$, denote the associated random number which is generated for each of the edges of the subdivision. Given a sample $\set{\rho_i\mid 1\le i\le 12a}$, let $M$ denote the set of indices of edges which are prefix minima or suffix minima in the sequence $(\rho_1,\ldots,\rho_{12a})$. For the purpose of our selection algorithm, we only need to generate set $M$ and the  $\rho_i$-values associated with $i\in M$. We start by generating those elements of $M$ which are prefix minima. By a slight abuse of notation, to initialise the process, we set $i_0=1$ and
generate $\rho_1$ uniformly at random in $[0,1]$, $\rho_1 := \rand(0,1)$. 
For successive $j = 0,1,2, \ldots$, we then generate the index $i_{j+1}$ of the first edge of $M$ after the edge with index $i_j$ (which is also the first index of $M$ having value less than $\rho_{i_j}$). As $i_{j+1}-i_j$ follows a geometric distribution with parameter $p = \rho_{i_j}$, this can be done in constant time by setting $i_{j+1} :=i_j+\floor{\frac{\log \rand(0,1) }{\log (1-p)}}$ (see e.g. \cite{devroye}). We then generate $\rho_{i_{j+1}}$ uniformly in $[0,\rho_{i_j}]$. We generate in this way indices $i_1,\ldots,i_x$, until for some $x \in \N$ we reach the $12a$ bound (i.e. $i_{x+1}>12a$). We then proceed similarly in reverse order for edges with index greater than $i_x$, to generate those edges with suffix minimal value (note that this time we have to sample values greater than $\rho_{i_x}$, rather than greater than $0$, and adapt all ranges accordingly, for consistency with the choice of prefix minima). In this way, we perform $O(\log a)$ constant-time sampling operations per edge of length $a$, obtaining $O(\log a)$ values (together with their positions) per edge, in expectation (and also w.h.p.\ with respect to $a$). Since random choices made for all edges of the original graph are independent, by a quick Chernoff bound, the total amortized sampling time over the whole graph is $O(m \log C)$, w.h.p., where $C$ denotes the maximum length of an edge. We remark that when constructing a hub set for a subset of nodes, each node only relies on the random choices made in its shortest-path tree, which can be evaluated in time $O(n \log C)$, w.h.p.


Our selection algorithm of the edge with minimal value in the middle window for a pair $u,v$ will necessarily select an edge that we have generated when the window contains a (real) edge extremity. Each time a virtual unweighted edge is selected as a hub, we indeed select the real edge it belongs to. We also have to manage the special case where the middle window entirely falls inside a long edge. In that case, the long edge is selected as hub.

The computation of the hub set is then a matter of scanning the tree by non-decreasing distance $r$ of generated vertices while maintaining a sliding middle window for each branch reaching distance $r$.
Using a balanced binary search tree per window for storing the virtual edges it contains, we obtain the hub set in $O(n\log C \log (n\log C))$ time.
\laurent{Better analysis by bounding nb of edges in each tree ? $\log D$ instead of $n\log C$ ?}
Distance labels can thus be computed in expected $O(nm+n^2\log C(\log n + \log \log C))$ time. Note that each of the $n$ labels can be computed independently (e.g. in parallel) in $O(m +n \log C (\log n + \log \log C))$ time per label, as long as randomness is shared (e.g. using random generators with same seeds).

\laurent{note of Laurent written by Adrian: why not use the reach algorithm as a heuristic for limiting (pruning) the Dijsktra tree when computing skeleton tree?}

\section{Generalizing the Definition of a Skeleton}
\label{sec:generalizations}

The definition of a skeleton, and the corresponding notion of skeleton dimension can be generalized in two ways, by using a different distinct distance metric to compute reach of a point in a tree, as well as by modifying the threshold value of reach required to retain a point in the skeleton.

\paragraph{Using two metrics.}
Suppose the graph $G$ is associated with two non-negative length functions, $\ell$ and $\ell'$, of its edges. For example, in road networks, one can typically consider travel time $\ell_t$ and geographic distance $\ell_d$, resulting in time and distance metrics, respectively. Another metric which may be of interest is hop count, corresponding to the constant function $\ell_h=1$ (i.e. $\ell_h(uv)=1$ for each edge $uv$).
Once a shortest path tree $T_u$ for node $u$ and its geometric realization $\T_u$ has been computed according to length function $\ell$, distance and reach within $\T_u$ can be computed according to $\ell'$.
Formally, extending the definition from Section~\ref{sec:skeldef}, the \emph{skeleton $T_u^{*\ell'}$} of $T_u$ is then defined as the subtree of $\T_u$ induced by $\{v\in V(\T_u)\mid \reach_{\T_u}^{\ell'}(v) \geq \frac{1}{2} d_{\T_u}^{\ell'}(u,v)\}$, where $d_{\T_u}^{\ell'}$ and $\reach_{\T_u}^{\ell'}$ denote distance and reach with respect to $\ell'$. The skeleton dimension of $G$ is then $k^{\ell'}=\max_{u\in V(G)}\width(T_u^{*\ell'})$. The advantage of this approach is that it sometimes results in smaller skeleton dimension (depending on the choice of metric $\ell'$), without affecting the correctness of any of the hub labeling schemes designed in this paper. As an example, the respective skeleton dimensions of the 9-th DIMACS New York graph~\cite{dimacs9th} for different choices of metrics turn out to be: $k^{\ell_t}=73$, $k^{\ell_d}=66$, and $k^{\ell_h}=56$, where $\ell_t$, $\ell_d$, and $\ell_h$ denote travel-time, geographic-distance, and hop-count length functions, respectively (considering shortest path trees for the metric $\ell_t$ in all three cases). 

We remark that a similar phenomenon, also taking advantage of two metrics, was observed and used in the reach-pruning approach~\cite{reach}. 

\paragraph{Modifying reach threshold.}
The choice of a reach threshold of $\frac{1}{2}$ in the definition of skeleton is arbitrary. 
Indeed, for any fixed $\a >0$, we can define the \emph{skeleton $T_u^{\a *}$} as the subtree of $\T_u$ induced by $\{v\in V(\T_u)\mid \reach_{\T_u}(v) \geq \a\, d_{\T_u}(u,v)\}$, and the skeleton dimension $k_\a$ of $G$ is given as $k_\a=\max_{u\in V(G)}\width(T_u^{\a*})$. The values of skeleton dimension for different values $\a$ and $\b$ of the reach threshold are related to each other by the following Proposition.

\begin{proposition}\label{pro:alphabeta}
For two constants $\a<\b$, the following bounds hold:
\[
k_\b \le
k_\a \le k_\b k_{\frac{\b+1}{\b/\a-1}}.
\]
\end{proposition}
\begin{proof}
The first relation is immediate since $T_u^{\b*}$ is a subtree of $T_u^{\a*}$ for $\b \ge \a$. The second relation is obtained by observing that $\cut_r(T_u^{\a*}) \subseteq \bigcup_{v\in \cut_{r'}(T_u^{\b*})} \cut_{r-r'}(T_v^{\b'*})$, with $r'=\frac{1+\a}{1+\b}$ and $\b'=\frac{\b+1}{\b/\a-1}$. Indeed,
for $w\in \cut_r(T_u^{\a*})$, we can consider in $\T_u$ the point $v$ at distance $r'$ from $u$ on the branch leading to $w$. The reach of $w$ in $\T_u$ is then at least
$r-r'+\reach_{\T_u}(v)\ge r-r'+\a r \ge \b r'$ for $(1+\a)r\ge (1+\b)r'$ and
$v$ thus belongs to $T_u^{\b*}$ for $r'=\frac{1+\a}{1+\b}$. Moreover, $w$ is at distance $r-r'$ from $v$ in $\T_v$ and has reach at least $\a r$, implying
$w\in V(T_v^{\b'*})$ for $\b'\le \frac{\a r}{r-r'}$, which is the case for $\b'=\frac{\b+1}{\b/\a-1}$ when $r'=\frac{1+\a}{1+\b}$.
\end{proof}
For $\beta=1$ and $\a<1$, the second relation of the above Proposition gives 
$k_\a \le k_1 k_{\frac{2\a}{1-\a}}$, which also implies that
$k_1\le k_\a\le k_1^2$ for $\a\ge 1/3$. More generally, we can derive the following bounds by repeatedly applying Proposition~\ref{pro:alphabeta} for $\b = 1$:
\[
k_1\le k_\a \le k_1^{\ceil{\log(1+1/\a)}}
\quad\mbox{ for }\quad
\a < 1\mbox{.}
\]
This shows that for a given graph, skeleton dimension $k_\a$ grows at most polynomially with $\frac{1}{\a}$.

\bigskip
Naturally, one can also apply both of the above-described generalizations together, obtaining a new skeleton dimension parameter $k_\a^{\ell'}$ with reach metric $\ell'$ and reach threshold $\a$. All the results of the paper about hub labelings and their computation in graphs with low skeleton dimension can be easily generalized to use $k_\a^{\ell'}$ instead of $k$, as long as $\a < 1$ (ensuring that any two skeleton trees $T_u^{\ell'*\a}$ and $T_v^{\ell'*\a}$ share a constant fraction of the $u-v$ shortest path). The particular choice of $k=k_{1/2}^\ell$ was made with the objective of clarity, and also on account of the simple relationship between $k_{1/2}^\ell$ and highway dimension.

\section{Conclusion} \label{sec:conclusion}

In this paper, we have proposed skeleton dimension as a measure of the network's amenability to shortest path schemes based on hub/transit nodes. We intend it as a parameter which is easy to describe and can be computed efficiently. Computations of hub sets based on skeleton dimension allow each node to individually and efficiently define its own hub set, subject only to a universal choice of random id-s. Such a construction is always correct, and gives small hub sets w.h.p. We remark that in a weighted network each node can compute its own appropriate labeling in $O(m+n\log C(\log n + \log \log C))$ time, where $C$ is the length of the longest integer weight in the network. The definition of hub sets, and the obtained bounds on their size, hold both for undirected and directed graphs. For directed graphs, skeleton dimension appears to be a parameter which is more directly usable than highway dimension.

Possible extensions of skeleton dimension, discussed in Section~\ref{sec:generalizations}, include variants of skeleton dimension with other values of reach threshold, as well as skeleton dimension defined using two separate distance metrics in the graph: one corresponding to the needs of the shortest path queries (used to construct shortest path trees), and another, potentially independent metric used internally in the computation of hub labelings, chosen so as to empirically minimize the width of the skeleton. When studying average-case parameters of a network, the integrated skeleton dimension given by~\eqref{eq:isk_def} (as well as its natural generalizations to weighted graphs) appear to be a natural parameter, which may be related to that of average highway dimension~\cite{hwVC}. We could also use the integrated skeleton dimension averaged over all nodes to get an even more accurate bound on average label size.

Finally, we remark on the interplay between skeleton and highway dimension. Skeleton dimension is always not greater than geometric highway dimension. We have also shown a clear case of separation in a weighted Manhattan-type network, where skeleton dimension is asymptotically much smaller than (geometric) highway dimension.

We remark that skeleton dimension appears particularly worthy of further theoretical study in the context of scale-free models of random graphs (cf.\ e.g.~\cite{Aldous28052013} for a discussion in the context of highway dimension and reach). For geometric percolation graphs, skeleton dimension displays a close link with the coalescence exponent for geodesics. Consequently, it may be easier to show rigorous theoretical bounds for skeleton dimension than for highway dimension.

\section*{Acknowledgment}

The authors thank Przemek Uzna\'nski and Olivier Marty for inspiring discussions on closely related problems. We also thank PU and Zuzanna Kosowska-Stamirowska for their help with the figures.

\bibliographystyle{plain}
\bibliography{paper}

\end{document}